\documentclass{article}
\usepackage[a4paper, total={6in, 9in}]{geometry}
\usepackage{graphicx} 
\usepackage{amsmath, amssymb, amsthm, amsfonts}
\usepackage[pdftex,dvipsnames]{xcolor}  
\usepackage{epstopdf}
\usepackage{subcaption}
\usepackage{hyperref}
\hypersetup{colorlinks = true,linkcolor=blue}
\numberwithin{figure}{section}
\newcommand{\nm}{\noalign{\smallskip}}
\newcommand{\ds}{\displaystyle}
\usepackage{soul}

\title{Nonlinear subwavelength resonances in three dimensions}

\date{}

\usepackage[numbers]{natbib}
\usepackage{xargs}                      
\usepackage[pdftex,dvipsnames]{xcolor}  
\usepackage[colorinlistoftodos,prependcaption,textsize=tiny]{todonotes}
\newcommandx{\unsure}[2][1=]{\todo[linecolor=red,backgroundcolor=red!25,bordercolor=red,#1]{#2}}
\newcommandx{\refrequest}[2][1=]{\todo[linecolor=blue,backgroundcolor=blue!25,bordercolor=blue,#1]{#2}}
\newcommandx{\note}[2][1=]{\todo[linecolor=OliveGreen,backgroundcolor=OliveGreen!25,bordercolor=OliveGreen,#1]{#2}}
\newcommandx{\improvement}[2][1=]{\todo[linecolor=Plum,backgroundcolor=Plum!25,bordercolor=Plum,#1]{#2}}
\newcommandx{\thiswillnotshow}[2][1=]{\todo[disable,#1]{#2}}

\newtheorem{theorem}{Theorem}[section]
\newtheorem{prop}[theorem]{Proposition}
\newtheorem{definition}[theorem]{Definition}
\newtheorem{lemma}[theorem]{Lemma}
\newtheorem{remark}[theorem]{Remark}
\newtheorem{corollary}[theorem]{Corollary}

\DeclareMathOperator{\dvol}{dvol}
\DeclareMathOperator{\Span}{span}
\DeclareMathOperator{\dsig}{d\sigma}
\DeclareMathOperator{\Capa}{cap}
\DeclareMathOperator{\resdom}{{\mathcal{D}}}
\DeclareMathOperator{\R3}{{\mathbb{R}^3}}
\newcommand*{\cres}{{c_r}}
\newcommand*{\cresj}{{c_r^j}}
\newcommand*{\cback}{{c_0}}

\newcommand*{\kappaback}{\kappa_0}
\newcommand*{\rhoback}{\rho_0}
\newcommand*{\kappares}{\kappa_r}
\newcommand*{\rhores}{\rho_r}
\newcommand*{\V}{V^{vol}}

\DeclareMathOperator{\Id}{Id}

\begin{document}

\author{Habib Ammari\thanks{\footnotesize Department of Mathematics, ETH Z\"urich, R\"amistrasse 101, CH-8092 Z\"urich, Switzerland (habib.ammari@math.ethz.ch, thea.kosche@sam.math.ethz.ch).} \and Thea Kosche\footnotemark[1]}

\maketitle

\begin{abstract} In this paper, we consider the resonance problem for the cubic nonlinear Helmholtz equation in the subwavelength regime.  We derive a discrete model for approximating  the subwavelength resonances of finite systems of high-contrast resonators  with Kerr-type  nonlinearities. Our discrete formulation is valid in both weak and strong nonlinear regimes. Compared to the linear formulation, it characterizes {the experimentally observed
extra eigenmodes that are induced by the nonlinearities.}  
\end{abstract}

\bigskip

 \par

\bigskip

\\

\tableofcontents

\section{Introduction and problem formulation}\label{section:intro}

In nanophotonics and nanophononics, the ideal artificial material shall, at subwavelength scales, confine and guide waves, and be able to modify the phase and enhance the amplitude of incoming waves. 
In order to do so, a class of nonlinear artificial materials consisting of nonlinear subwavelength high-contrast resonator building blocks has been introduced. This class of materials has revived the exciting prospect of realizing the long-standing goal of robust transport of waves at subwavelength scales and their amplification in only one direction, opening avenues to channel energy and information transport in quantum computing; see, \emph{e.g.}, \cite{soliton3,rev-acoustics,soliton1,active,soliton2,nonreciprocal}. 

Although very significant advances in experimental and numerical modeling have been achieved in the field of nonlinear subwavelength physics during recent decades, very little is known from a mathematical point of view. To our knowledge, the properties of nonlinear subwavelength resonances have not yet been  analyzed in the mathematical literature. In contrast to the linear response regime studied in depth in \cite{cbms,review}, there is a clear lack of deep understanding of the theory of nonlinear subwavelength resonances.

In this work, we lie the basis for nonlinear subwavelength physics. The focus is to develop a theory of \emph{nonlinear subwavelength resonances}. This is accomplished by using ideas from \cite{feppon2024subwavelength} in the context of nonlinear systems. It is worth emphasizing that the nonlinearity renders the nonlinear problem much more difficult than the linear one. The total number of subwavelength resonances may be larger than 
the total number of subwavelength resonators and depends on the strength of 
the nonlinearities as opposed to the linear systems where the total number of subwavelength resonances is exactly given by the number of resonators. Some extra-eigenmodes  may exist due to strong enough nonlinearities. These eigenmodes are of profound significance to nonlinear subwavelength physics. They have been experimentally observed in systems of high-contrast resonators where the nonlinearity is of Kerr type \cite{soliton1}.

Our main contribution in this work is to provide for the first time a discrete formulation for approximating the subwavelength resonances of finite systems of high-contrast resonators  with Kerr-type  nonlinearities, which is valid in both weak and strong nonlinear regimes. Compared to the linear formulation, our formulation here characterizes the {extra eigenmodes induced by the nonlinearities}.  Our main tool for building such a discrete framework is a generalization of the Dirichlet-to-Neumann approach for subwavelength resonances from \cite{feppon2024subwavelength} to the nonlinear setting combined with asymptotic analysis of solutions to nonlinear resonance problems for the Helmholtz equation (with respect to the contrast in the material parameters of the resonators).

Concretely, let $\resdom \subset \R3$ be a smooth bounded domain with connected components $B_1,\ldots,B_N$ that have connected boundaries. Let $\kappaback \in \mathbb{R}$, $\rhoback \in \mathbb{R}$ and let $\kappares: \resdom \rightarrow \mathbb{R}$, $\rhores: \resdom \rightarrow \mathbb{R}$ be piecewise constant functions.  We consider the following system of equations:

\begin{equation}\label{eq:wave_eq_non-linear}
    \left\{
    \begin{array}{rll}
   \ds     \frac{1}{\kappaback}\frac{\partial^2 u}{\partial t^2}  - \frac{1}{\rhoback}\Delta u &\hspace{-5pt}= 0 &\text{in } \mathbb{R} \times \R3\setminus\overline{\resdom},\\
   \nm
     \ds    \frac{1}{\kappares(x)}\frac{\partial^2 u}{\partial t^2}  - \frac{1}{\rhores(x)}\Delta u &\hspace{-5pt}= g(u) &\text{in } \mathbb{R} \times \resdom,\\
     \nm
      \ds  \frac{1}{\rhoback}\left.\frac{\partial u}{\partial \nu}\right\vert_+ &\ds\hspace{-5pt}= \frac{1}{\rhores(x)}\left.\frac{\partial u}{\partial \nu}\right\vert_- &\text{on } \mathbb{R} \times \partial\resdom,\\
        \nm
      \ds  \left.u\right\vert_+ &\hspace{-5pt}= \left.u\right\vert_- &\text{on } \mathbb{R} \times \partial\resdom,
    \end{array}
    \right.
\end{equation}
where   $u$ is subject to an outgoing radiation condition, $\nu$ is the outward normal to $\resdom$, and $g$ is a nonlinear operator. In this article, we consider mainly
$$ g(u) = {\frac{\beta}{\kappares}} \left\vert \frac{\partial u}{\partial t}\right\vert^2\frac{\partial u}{\partial t} \quad \text{and} \quad  g(u) = 
\frac{\beta}{\kappares} \left(\frac{\partial u}{\partial t}\right)^2\overline{u},$$
with $\beta \in \mathbb{C}$.
For given operator $g$, we say that $\omega \in \mathbb{C}$ is a \emph{resonance} of the system \eqref{eq:wave_eq_non-linear}, if there exists a non-zero $u: \R3 \rightarrow \mathbb{C}$ such that
    $$ (t,x) \in \mathbb{R}\times\R3 \quad \longmapsto \quad \exp(i\omega t)u(x) \in \mathbb{C} $$
is a solution to \eqref{eq:wave_eq_non-linear}. 

In the linear case, when $g(u) = 0$, resonances are precisely those $\omega \in \mathbb{C}$ such that the scattering problem 
\begin{equation}\label{eq:helmholtz_linear}
    \left\{
    \begin{array}{rll}
     \ds   \frac{1}{\rhoback}\Delta u + \frac{\omega^2}{\kappaback} u &\hspace{-5pt}= 0 &\text{in } \R3\setminus\overline{\resdom},\\
     \nm \ds
        \frac{1}{\rhores(x)}\Delta u + \frac{\omega^2}{\kappares(x)} u &\hspace{-5pt}= 0 &\text{in } \resdom,\\
        \nm\ds
        \frac{1}{\rhoback} \left.\frac{\partial u}{\partial \nu}\right\vert_+  &\ds\hspace{-5pt}= \frac{1}{\rhores(x)}\left.\frac{\partial u}{\partial \nu}\right\vert_- &\text{on } \partial\resdom,\\
        \nm \ds
        \left. u \right\vert_+ &\hspace{-5pt}= \left.u\right\vert_- &\text{on } \partial\resdom,
    \end{array}
    \right.
\end{equation}
where $u - u_{in}$ is subject to an outgoing radiation condition 
with the incident wave $u_{in}$ satisfying $ \frac{1}{\rhoback}\Delta u_{in} + \frac{\omega^2}{\kappaback} u_{in} =0$ in $\R3$, does not have a unique solution or, equivalently, is degenerate.
We rewrite the material parameters of \eqref{eq:wave_eq_non-linear} and \eqref{eq:helmholtz_linear} as follows:
\[
\cback := \sqrt{\frac{\kappaback}{\rhoback}}, \quad
\cres(x) :=\sqrt{\frac{\kappares(x)}{\rhores(x)}}, \quad
\delta(x) := \frac{\rhores(x)}{\rhoback}.
\]
For simplicity of notation, we  assume that $\delta(x) \equiv \delta$ and $ c_r(x) \equiv c_r$ are constant across the different connected components of $\resdom$.
We consider the regime when $\delta \rightarrow 0$ and study those resonances which satisfy $\omega(\delta) \rightarrow 0$ as $\delta \rightarrow 0$. 
To this end, we assume that $\cback$ and $\cres$ are constant as $\delta \rightarrow 0$.
\begin{definition}[Subwavelength resonance]\label{def:subw_res}
    Let $\omega: (0,\epsilon) \rightarrow \mathbb{C}$ be a continuous function for some $\epsilon>0$. If $\omega(\delta)$ is a resonance of the system \eqref{eq:helmholtz_linear} for every $\delta \in (0,\epsilon)$ and if $\omega(\delta) \rightarrow 0$ as $\delta \rightarrow 0$, then we say that $\omega(\delta)$ is a \emph{(linear) subwavelength resonance}.
\end{definition}

Linear subwavelength resonances have been thoroughly studied since the original paper \cite{ihp}; see, for instance, \cite{cbms,review}.  The leading-order behavior of subwavelength resonances $\omega(\delta)$ as $\delta \rightarrow 0$ arecharacterized by the so-called generalized capacitance matrix $\Capa^{gen}(\resdom) \in \mathbb{R}^{N \times N}$. In fact, one has the following theorem.
\begin{theorem}[Fundamental theorem of subwavelength physics \cite{review}]\label{thm:fundamental}
    Let $\resdom \subset \mathbb{R}^3$ be a smooth bounded domain with well-separated connected components $B_1,\ldots,B_N$ which are supposed to have connected boundaries. Then there exists a matrix $\Capa^{gen}(\resdom) \in \mathbb{R}^{N \times N}$ depending on $\resdom$ and $c_r$ and there exists an $\epsilon > 0$, such that the following correspondence holds
    \begin{equation*}
    \begin{array}{c}
    \left\{
        \begin{array}{cc}
             \omega(\delta) : (0,\epsilon) \longrightarrow \mathbb{C} ~\vert~ \omega \text{ is a subwavelength resonance}
        \end{array}
        \right\} \vspace{0.2cm} \\\tilde{\longleftrightarrow} \vspace{0.2cm} \\
        \left\{
        \begin{array}{cc}
             \lambda \in \mathbb{R} ~\vert~ \lambda \text{ is an eigenvalue of } \Capa^{gen}(\resdom)
        \end{array}
        \right\},
    \end{array}
    \end{equation*} 
    where $\omega(\delta)$ is associated to the eigenvalue $\lambda$ of $\Capa^{gen}(\resdom)$ such that
    $$ \lvert \omega(\delta) - \sqrt{\lambda}\sqrt{\delta}\rvert = O(\delta) \quad \text{ as } \quad \delta \longrightarrow 0. $$
\end{theorem}

This theorem has been proved using the Gohberg-Sigal theory and layer potential techniques. It sets the basis of many results in subwavelength physics. It has been reproduced using the different approach of Dirichlet-to-Neumann operator by Feppon and Ammari \cite{feppon2024subwavelength} for connected $\resdom = B_1$.
In this article, we will enlarge the approach of Dirichlet-to-Neumann operators to disconnected $\resdom = B_1 \cup \ldots \cup B_N$, reproducing Theorem \ref{thm:fundamental}, in Section \ref{section:linear}. Using similar techniques, we will then develop a theory of nonlinear subwavelength resonances in Section \ref{section:non-linear}.

Let $H^1(\resdom)$ denote the usual Sobolev space of square-integrable functions whose weak derivative is square-integrable in $\resdom$. We also denote by $H^{\frac{1}{2}}(\partial\resdom)$ the set of traces on $\partial \resdom$ of the functions in $H^1(\resdom)$  and let $H^{-\frac{1}{2}}(\partial\resdom)$ be the dual of $H^{\frac{1}{2}}(\partial\resdom)$.

Concretely, the approach is as follows. 
Let $\mathcal{T}^k_{\resdom} : H^{\frac{1}{2}}(\partial\resdom) \rightarrow H^{-\frac{1}{2}}(\partial\resdom)$ be the Dirichlet-to-Neumann operator associated to the exterior Helmholtz problem
\begin{equation*}
        (\Delta + k^2)u = 0 \quad \text{on }\mathbb{R}^3\setminus\overline{\resdom},
\end{equation*}
subject to an outgoing radiation condition.
This allows to rewrite the system \eqref{eq:helmholtz_linear} in terms of $\mathcal{T}_{\resdom}^{\frac{\omega}{\cback}}$ and $u|_{\resdom}$ as follows:
\begin{equation*}
    \left\{
    \begin{array}{rll}
     \ds   \Delta u + \frac{\omega^2}{\cres^2} u &\hspace{-5pt}= 0 &\text{in } \resdom,\\
        \nm 
        \ds
        \frac{\partial u}{\partial \nu} &\hspace{-5pt}=  \ds \delta\left(\mathcal{T}_{\resdom}^{\frac{\omega}{\cback}}[u - u_{in}] + \frac{\partial u_{in}}{\partial \nu}\right) &\text{on } \partial\resdom.
    \end{array}
    \right.
\end{equation*}
Formulating this system variationally leads to the following system in terms of $u|_{\resdom}$ and $\omega$:
\begin{equation*}
    \begin{split}
        \int_{\resdom}\nabla u \cdot \nabla \overline{v} \dvol - \frac{\omega^2}{\cres^2}\int_{\resdom}u\overline{v}\dvol - \delta \int_{\partial \resdom} \mathcal{T}_{\resdom}^{\frac{\omega}{\cback}}[u]\overline{v}\dsig = \delta \int_{\partial\resdom}\left(\frac{\partial u_{in}}{\partial \nu} - \mathcal{T}^{\frac{\omega}{v_0}}[u_{in}]\right)\dsig,\\ \quad \forall v \in H^1(\resdom).
    \end{split}    
\end{equation*}
Then $\omega$ is a resonance whenever this system is degenerate and does not have a unique solution $u|_{\resdom}$. Or equivalently $\omega$ is a resonance whenever this system admits a non-zero solution $u|_{\resdom}$ for $u_{in} \equiv 0$. It was observed in \cite{feppon2024subwavelength} that one can facilitate the analysis and transform this \emph{indirect} problem of finding characteristic points $\omega$ of the above system into a \emph{direct} problem, in the case of a connected domain $\resdom$. By adding $\int_{\resdom} u \dvol \int_{\resdom} \overline{v}\dvol$ to the bilinear form on the left and $\int_{\resdom} \overline{v}\dvol$ to the linear functional on the right one obtains
\begin{equation}\label{eq:var_prob_con}
    \begin{split}
        \int_{\resdom}\nabla u \cdot \nabla \overline{v} \dvol + \int_{\resdom} u \dvol \int_{\resdom} \overline{v}\dvol - \frac{\omega^2}{\cres^2}\int_{\resdom}u\overline{v}\dvol - \delta \int_{\partial \resdom} \mathcal{T}_{\resdom}^{\frac{\omega}{\cback}}[u]\overline{v}\dsig
        = \int_{\resdom} \overline{v}\dvol,\\ \quad \forall v \in H^1(\resdom)
    \end{split}    
\end{equation}
for $u_{in} \equiv 0$. In the case of  sufficiently small $\delta \in \mathbb{R}_{>0}$ and sufficiently small $\omega \in \mathbb{C}$, the bilinear form on the left remains coercive as a small perturbation of the continuous coercive bilinear form 
\begin{equation*}
    \begin{split}
        a_{\omega,\delta}:~(u,v)\in H^1(\resdom)\times H^1(\resdom) \longmapsto \int_{\resdom}\nabla u \cdot \nabla \overline{v} \dvol + \int_{\resdom} u \dvol \int_{\resdom} \overline{v}\dvol
    \end{split}    
\end{equation*}
does not change its coercivity \cite{feppon2024subwavelength}.
Then the resonances of \eqref{eq:helmholtz_linear} are characterized as follows \cite{feppon2024subwavelength}: $\omega$ is a resonance if and only if the unique solution $u$ to the variational problem \eqref{eq:var_prob_con} satisfies 
\begin{equation}\label{eq:con_u_1}
    \int_{\resdom} u \dvol = 1.
\end{equation}
In the subwavelength regime, $\delta \rightarrow 0$, it is possible to solve the variational problem \eqref{eq:var_prob_con} with the Ansatz 
\begin{equation*}
    u_{\omega,\delta} = \sum_{p \geq 0, k \geq 0} \omega^p \delta^k u_{p,k}.
\end{equation*}
Computing $\ds \int_{\resdom} u_{\omega,\delta} \dvol$ and equating 
\begin{equation*}
    \sum_{p \geq 0, k \geq 0} \omega^p \delta^k \int_{\resdom} u_{p,k} \dvol = 1
\end{equation*}
allows to solve for $\omega(\delta)$ and to determine asymptotically the subwavelength resonances in the setting of a connected domain $\resdom$, see \cite[Equation (3.21)]{feppon2024subwavelength}.

In this work we generalize the Dirichlet-to-Neumann operator approach to disconnected resonator domains $\resdom = B_1 \cup \ldots \cup B_N$ (in Section \ref{section:linear}) and to nonlinear systems (in Section \ref{section:non-linear}). The main difference is the following. Instead of adding $\int_{\resdom} u \dvol \int_{\resdom}\overline{v}\dvol$ to obtain a coercive bilinear form similar to \eqref{eq:var_prob_con}, one has to account for the disconnectedness of $\resdom = B_1 \cup \ldots \cup B_N$ and add $\sum_{i = 1}^N \int_{B_i}u \dvol \int_{B_i}\overline{v}\dvol$ instead. In that case one obtains instead of \eqref{eq:con_u_1} a system of equations accounting for the integral of $u$ on the different connected components $B_i$ of $\resdom$, the details can be found in Subsection \ref{subsec:var_form_lin}.
Similarly to the connected case, one can make the Ansatz $u_{\omega,\delta} = \sum_{p \geq 0, k \geq 0} \omega^p \delta^k u_{p,k}$ and solve the system of equations asymptotically for subwavelength resonances $\omega(\delta)$. This is presented in Subsection \ref{subsec:char_res_lin}.

Suppose that $g(u) = \beta\lvert \frac{\partial u}{\partial t}\rvert^2\frac{\partial u}{\partial t}$. Then we study the following system of equations:
\begin{equation}\label{eq:wave_eq_nonlinear}
    \left\{
    \begin{array}{rll}
      \ds  \frac{1}{\kappaback}\frac{\partial^2 u}{\partial t^2}  - \frac{1}{\rhoback}\Delta u &\hspace{-5pt}= 0 &\text{in } \mathbb{R} \times \R3\setminus\overline{\resdom},\\
      \nm \ds
        \frac{1}{\kappares(x)}\frac{\partial^2 u}{\partial t^2}  - \frac{1}{\rhores(x)}\Delta u &\ds\hspace{-5pt}= \beta\left\vert \frac{\partial u}{\partial t}\right\vert^2\frac{\partial u}{\partial t} &\text{in } \mathbb{R} \times \resdom,\\
        \nm \ds
        \frac{1}{\rhoback}\left.\frac{\partial u}{\partial \nu}\right\vert_+ &\ds\hspace{-5pt}= \frac{1}{\rhores(x)}\left.\frac{\partial u}{\partial \nu}\right\vert_- &\text{on } \mathbb{R} \times \partial\resdom,\\
        \nm \ds
        \left.u\right\vert_+ &\ds\hspace{-5pt}= \left.u\right\vert_- &\text{on } \mathbb{R} \times \partial\resdom,
    \end{array}
    \right.
\end{equation}
subject to an outgoing radiation condition, 
where we consider the Kerr-like nonlinearity $\beta\lvert \frac{\partial u}{\partial t}\rvert^2\frac{\partial u}{\partial t}$ inside the resonator domain $\resdom$. 
As in the linear case, we will say that $\omega \in \mathbb{C}$ is a \emph{resonance} of the system \eqref{eq:wave_eq_nonlinear} whenever it admits a non-zero solution of the form $\exp(i \omega t)u(x)$. This allows to consider the following nonlinear Helmholtz resonance problem:
\begin{equation}\label{eq:helmholtz_non-linear}
    \left\{
    \begin{array}{rll}
     \ds   \frac{1}{\rhoback}\Delta u + \frac{\omega^2}{\kappaback} u &\ds\hspace{-5pt}= 0 &\ds\text{in } \R3\setminus\overline{\resdom},\\
     \nm \ds
        \frac{1}{\rhores(x)}\Delta u + \frac{\omega^2}{\kappares(x)} u &\ds\hspace{-5pt}= {-}i\beta\lvert \omega\rvert^2\omega \lvert u\rvert^2u &\ds\text{in } \resdom,\\
        \nm \ds
        \frac{1}{\rhoback}\left.\frac{\partial u}{\partial \nu}\right\vert_+ & \ds\hspace{-5pt}= \frac{1}{\rhores(x)}\left.\frac{\partial u}{\partial \nu}\right\vert_- &\ds\text{on } \partial\resdom,\\
        \nm \ds
        u\vert_+ &\ds\hspace{-5pt}= \left.u\right\vert_- &\text{on } \partial\resdom,
    \end{array}
    \right.
\end{equation}
subject to an outgoing radiation condition.
As for \eqref{eq:helmholtz_linear}, we will say that $\omega(\delta)$ is a \emph{subwavelength resonance}, whenever it is a continuous map $\omega: (0,\epsilon) \rightarrow \mathbb{C}$ for some $\epsilon > 0$ such that $\omega(\delta)$ is a resonance of \eqref{eq:helmholtz_non-linear} for all $\delta \in (0,\epsilon)$ and such that $\omega(\delta) \rightarrow 0$ as $\delta \rightarrow 0$.

The approach to subwavelength resonances in the nonlinear setting is identical to the approach in the linear setting in Section \ref{section:linear}, with the exception that one has to consider a nonlinear form $(u,v) \mapsto \tilde{a}_{\omega,\delta}(u,v)$ for the variational formulation of \eqref{eq:helmholtz_non-linear}, see Subsection \ref{subsec:var_form_non-lin} and that one has to keep track of the amplitudes of the solutions to the variational problem, see Theorem \ref{thm:var_nonlin}.
The subwavelength resonances are then determined in the same way as for the linear setting, by expanding in $\omega$ and $\delta$ and determining $\omega(\delta)$ such that the characterization of a subwavelength frequency is satisfied; see Subsection \ref{subsec:char_res_non-lin}. 

This article is organized as follows. In Section \ref{section:pre} we set the notation and technical prerequisites which will be useful in the later sections. Section \ref{section:linear} is dedicated to the characterization of subwavelength resonances of linear Helmholtz systems with disconnected resonator domain $\resdom$. It is divided into Subsection \ref{subsec:var_form_lin} where the partial differential equation (PDE) system is formulated variationally and the resonances are variationally characterized  and into Subsection \ref{subsec:char_res_lin} where the asymptotics of subwavelength resonances are computed. Section \ref{section:non-linear} then builds on top of these techniques and follows a similar structure. First in Subsection \ref{subsec:var_form_non-lin} the system of nonlinear PDE's is formulated variationally and resonances are characterized. Then, in Subsection \ref{subsec:char_res_non-lin}, the asymptotics of nonlinear subwavelength resonances are derived.
In Section \ref{section:num} simulations of the leading-order asymptotics of the subwavelength resonances of a nonlinear system of dimers $\resdom = B_1 \cup B_2$  are presented, demonstrating the existence of more subwavelength resonances and more subwavelength solutions than in the linear case.
This article is concluded with Section \ref{section:outlook} discussing the final conclusions and an outlook for further research.


\section{Prerequisites and notation}\label{section:pre}
In what follows, $\resdom \subset \mathbb{R}^3$ is a smooth bounded domain with connected components $B_1,\ldots, B_N$ that have connected boundaries. Their respective volume (w.r.t. Lebesgue measure) is denoted by $\lvert B_j \rvert$, their characteristic functions are denoted by $1_{B_j}$ and are given by
\begin{align*}
    1_{B_j}(x) = \begin{cases}
        1 \quad\text{if } x \in B_j,\\
        0 \quad\text{otherwise.}
    \end{cases}
\end{align*}
For $k \in \mathbb{C}$, the single layer potential  $\mathcal{S}^k_{\resdom}$ and the adjoint of the Neumann-Poincaré operator $\mathcal{K}^{k,*}_{\resdom}$ are given for $\phi \in H^{-\frac{1}{2}}(\partial\resdom)$ by 
\begin{align*}
    \mathcal{S}^k_{\resdom}[\phi](x) &:=
        \int_{\partial \resdom} G^k(x - y)\phi(y)\dsig(y), \quad x \in \partial \resdom,\\
        \nm
    \mathcal{K}^{k,*}_{\resdom}[\phi](x) &:= \int_{\partial \resdom} 
    \frac{\partial  G^k(x - y)}{\partial \nu(x)} \phi(y)\dsig(y), \quad x \in \partial \resdom,
\end{align*}
where $\nu(x)$ denotes the outward pointing normal vector to $\partial \resdom$ at $x \in \partial \resdom$ and the Green's function $G^k$ of the Helmholtz equation with radiation condition is given by $$G^k(x - y) := - \frac{\exp(ik\lvert x - y\rvert)}{4 \pi \lvert x - y \rvert}.$$

Let $\phi \in H^{\frac{1}{2}}(\partial \resdom)$. Recall that the following system:
\begin{equation*}
    \left\{\begin{array}{rl}
        \Delta u + k^2 u &= ~~ 0 \quad \text{ in } \mathbb{R}^3\setminus \overline{\resdom},\\
        \nm
        u &= ~~ \phi \quad \text{ on } \partial \resdom,
    \end{array}\right.
\end{equation*}
subject to an outgoing radiation condition, has a unique solution.
In that case, the Dirichlet-to-Neumann operator $\mathcal{T}^k_{\resdom}[\phi]$ is defined as 
$$ \mathcal{T}^{k}_{\resdom}[\phi](x) := \frac{\partial u}{\partial \nu} (x), \quad x \in \partial \resdom. $$ 

The Dirichlet-to-Neumann operator can be expanded into a type of generalized Taylor series in the space of bounded linear operators $ \mathcal{L}(H^{\frac{1}{2}}(\partial\resdom), H^{-\frac{1}{2}}(\partial\resdom))$ for $k$ in a neighborhood of $0 \in \mathbb{C}$.
The next results can be deduced from \cite[Proposition 3.1]{feppon2024subwavelength} and properties of the single layer potential  which, \emph{e.g.}, can be found in \cite{photonic}.

\begin{lemma}\label{lem:DtN}
    Let $k \in \mathbb{C}$ and suppose that $k^2$ is not a Dirichlet eigenvalue for $-\Delta$ in $\resdom$. Then, the following statements hold:
    \begin{enumerate}
        \item[(i)] The single layer potential  $\mathcal{S}^k_{\resdom}: H^{-\frac{1}{2}}(\partial\resdom) \rightarrow H^{\frac{1}{2}}(\partial\resdom)$ is invertible;
        \item[(ii)] The Dirichlet-to-Neumann operator $\mathcal{T}^k_{\resdom}: H^{\frac{1}{2}}(\partial \resdom) \rightarrow H^{-\frac{1}{2}}(\partial \resdom)$ is given by 
        $$ \mathcal{T}^k_{\resdom}[\phi] = \left(\frac{1}{2}\Id + \mathcal{K}^{k,*}_{\resdom} \right)\left( \mathcal{S}_{\resdom}^{k}\right)^{-1}[\phi], $$
        where $\Id$ denotes the identity operator;
        \item[(iii)] For $k \rightarrow 0 \in \mathbb{C}$, the Dirichlet-to-Neumann operator and the single layer potential  can be expanded as a convergent power series in the spaces of bounded linear operators $\mathcal{L}(H^{\frac{1}{2}}(\partial\resdom), H^{-\frac{1}{2}}(\partial\resdom))$ and $\mathcal{L}(H^{-\frac{1}{2}}(\partial\resdom), H^{\frac{1}{2}}(\partial\resdom))$, respectively. The following identities hold:
        $$ \mathcal{T}^k_{\resdom} = \sum_{n \geq 0} k^n \mathcal{T}_n, \quad \mathcal{S}^k_{\resdom} = \sum_{n \geq 0} k^n \mathcal{S}_n, $$
        where
        \begin{align*}
            \mathcal{T}_0 = \mathcal{T}^{0}_{\resdom}, \quad\quad \mathcal{S}_0 = \mathcal{S}^{0}_{\resdom} \quad \text{ and } \quad
            \mathcal{T}_1[\phi] = -\mathcal{T}^0_{\resdom} \circ \mathcal{S}_1 \circ (\mathcal{S}^0_{\resdom})^{-1}.
        \end{align*}
    \end{enumerate}
\end{lemma}
For this work, the most important part of this Lemma \ref{lem:DtN} will be item (iii). The formulas provided for $\mathcal{T}_0$ and $\mathcal{T}_1$ will be needed to determine the subwavelength resonances asymptotically in terms of $\delta$ in Subsection \ref{subsec:char_res_lin} and Subsection \ref{subsec:char_res_non-lin}.

In the introduction (Section \ref{section:intro}), the fundamental theorem of subwavelength physics was presented in Theorem \ref{thm:fundamental}. The statement is that the leading order asymptotics of subwavelength resonances (in the linear setting) correspond to the eigenvalues of a matrix $\Capa^{gen}(\resdom)$ called the generalized capacitance matrix.
 The following functions will be needed for its definition.

\begin{definition}\label{def:Phi}
    Let $j \in \{1, \ldots, N\}$. 
    Define the following functions $\Phi_j: \R3 \setminus \overline{\resdom} \rightarrow \mathbb{C}$ as the solution to the following system of equations:
    \begin{equation*}
        \left\{ \begin{array}{rll}
            \Delta \Phi_j &=~~ 0 & \text{in } \R3 \setminus \overline{\resdom},\\
            \Phi_j \vert_{\partial \resdom} &=~~ 1_{\partial B_j} &\text{ on } \partial \resdom,\\
            \Phi_j(x) &= \mathcal{O}(\lvert x \rvert^{-1}) &\text{ as } \lvert x \rvert \rightarrow \infty. 
        \end{array} \right.
    \end{equation*}
\end{definition}
With these objects in hand, it is now possible to formulate a precise definition of the capacitance matrix, which is as follows \cite{ammari2017plasmonic}.

\begin{definition}[Capacitance matrix]\label{def:Capa}
    Let $\resdom \subset \mathbb{R}^3$ be a smooth bounded domain with connected components $B_1,\ldots,B_N$ that are well separated and have connected boundaries. The \emph{capacitance matrix} $\Capa(\resdom) \in \mathbb{R}^{N \times N}$ associated to $\resdom$ is defined as
    \begin{equation*}
        \Capa(\resdom)_{l,j} = -\int_{\partial B_l} \frac{\partial \Phi_j}{\partial \nu} \dsig.
    \end{equation*}
    Let $\cresj$ denote
    the material parameter in $B_j$, for $j = 1,\ldots,N$, and let $\V$ and $V^{mat}$ be the diagonal matrices with the diagonal entries defined as
    \begin{equation*}
        (\V)_{j,j} = \frac{1}{\lvert B_j \rvert} \quad \text{ and } \quad (V^{mat})_{j,j} = (\cresj)^2.
    \end{equation*}
    The generalized capacitance matrix $\Capa^{gen}(\resdom) \in \mathbb{C}^{N \times N}$ associated to $\resdom$ and the material parameters $(\cresj)_j$ is defined as
    \begin{equation*}
        \Capa^{gen}(\resdom) =  V^{mat}\V \Capa(\resdom).
    \end{equation*}
\end{definition}

Note that if $\cresj \equiv \cres$ for all $j$, then 
    \begin{equation*}
        \Capa^{gen}(\resdom) = \cres^2 \V \Capa(\resdom).
    \end{equation*}

In Section \ref{section:linear}, we reproduce Theorem \ref{thm:fundamental} with the techniques of Dirichlet-to-Neumann operators and a variational approach.


\section{A Dirichlet-to-Neumann approach to subwavelength resonances}\label{section:linear}

This section is devoted to reproducing the fundamental theorem of subwavelength physics (Theorem \ref{thm:fundamental}) with a novel technique, first used for this purpose in \cite{feppon2024subwavelength}. The original derivation was based on layer potential techniques and the \emph{asymptotic} Gohberg-Sigal theory (as can be found, for example, in \cite{photonic}). In this work, we take a different approach which is based on Dirichlet-to-Neumann operators and variational arguments. To this end, we pose the following proposition which reformulates the system of equations \eqref{eq:helmholtz_linear} in terms of the Dirichlet-to-Neumann operator and as a variational problem in $H^1(\resdom)$.

\begin{prop}\label{prop:equiv_var_non-var}
    Let $u \in H^1(\resdom)$.  Then the following statements are equivalent:

    \begin{enumerate}
        \item[(i)] The function $u$ solves the system \eqref{eq:helmholtz_linear} for some outgoing $u|_{\mathbb{R}^3 \setminus \overline{\resdom}} \in H^1_{loc}(\mathbb{R}^3\setminus \overline{\resdom})$;
        \item[(ii)] The function $u$ satisfies 
        \begin{equation*}
            \left\{
            \begin{array}{rll}
            \ds    \Delta u + \frac{\omega^2}{\cres^2} u &\hspace{-5pt}= 0 &\text{in } \resdom,\\
                \nm \ds
                \frac{\partial u}{\partial \nu} &\hspace{-5pt}=  \ds\delta\left(\mathcal{T}_{\resdom}^{\frac{\omega}{\cback}}[u - u_{in}|_+] + \left.\frac{\partial u_{in}}{\partial \nu}\right|_+\right) &\text{on } \partial\resdom ;
            \end{array}
            \right.
        \end{equation*}
        \item[(iii)] The function $u$ solves the variational problem
            \begin{equation}\label{eq:helmholtz_var}
                    \begin{split}
                        \int_{\resdom}\nabla u \cdot \nabla \overline{v} \dvol - \frac{\omega^2}{\cres^2}\int_{\resdom}u\overline{v}\dvol - \delta \int_{\partial \resdom} \mathcal{T}_{\resdom}^{\frac{\omega}{\cback}}[u]\overline{v}\dsig = \delta \int_{\partial\resdom}\left(\left.\frac{\partial u_{in}}{\partial \nu}\right|_+ - \mathcal{T}_{\resdom}^{\frac{\omega}{\cback}}[u_{in}|_+]\right)\overline{v}\dsig,\\ \quad \forall v \in H^1(\resdom).
                    \end{split}    
            \end{equation}
    \end{enumerate}
\end{prop}

\begin{proof}
    Assume that $u$ solves the system \eqref{eq:helmholtz_linear} for some outgoing $u|_{\mathbb{R}^3 \setminus \overline{\resdom}} \in H^1_{loc}(\mathbb{R}^3\setminus \overline{\resdom})$. Then $\left.\left(u|_{\mathbb{R}^3 \setminus \overline{\resdom}}\right)\right|_+ = u|_- -u_{in}$ and $u_{\mathbb{R}^3 \setminus \overline{\resdom}}$ is outgoing and solves the exterior Helmholtz problem. Thus,  
    $$\left.\frac{\partial u|_{\mathbb{R}^3 \setminus \overline{\resdom}}}{\partial\nu}\right|_+ = \mathcal{T}^{\frac{\omega}{\cback}}_{\resdom}\left[\left.\left(u|_{\mathbb{R}^3 \setminus \overline{\resdom}}\right)\right|_+\right] = \mathcal{T}^{\frac{\omega}{\cback}}_{\resdom}[u|_- -u_{in}].$$
   Hence, $$ \left.\frac{\partial u}{\partial \nu}\right|_- = \delta\left(\mathcal{T}_{\resdom}^{\frac{\omega}{\cback}}[u|_- - u_{in}|_+] + \left.\frac{\partial u_{in}}{\partial \nu}\right|_+\right) $$ and (ii) follows. Clearly, if $u$ satisfies (ii), then there exists an outgoing $u|_{\mathbb{R}^3\setminus \resdom} \in H^1_{loc}(\mathbb{R}^3\setminus \overline{\resdom})$ such that
    $$ \left.\left(u|_{\mathbb{R}^3\setminus \resdom}\right)\right|_+ = u|_- - u_{in}|_+ \quad \text{and} \quad \left.\frac{\partial u|_{\mathbb{R}^3 \setminus \overline{\resdom}}}{\partial\nu}\right|_+ = \mathcal{T}^{\frac{\omega}{\cback}}_{\resdom}\left[ u|_- - u_{in}|_+ \right]. $$
    Suppose that $u \in H^1(\resdom)$ satisfies (ii) and let $v \in H^1(\resdom)$ be arbitrary, multiplying the first equation in (ii) with $\overline{v}$ and integrating by parts over $\resdom$  gives 
    \begin{equation*}
        \int_{\resdom} \left(\Delta u + \frac{\omega^2}{\cres^2}u\right) \overline{v}\dvol =  \int_{\resdom} -\nabla u \cdot \nabla \overline{v} + \frac{\omega^2}{\cres^2}u \overline{v}\dvol + \int_{\partial \resdom}\frac{\partial u}{\partial \nu} \overline{v}\dsig.
    \end{equation*}
    Inserting the second equation of (ii) for ${\partial u}/{\partial \nu}$ one obtains that $u$ solves the variational problem in (iii). Suppose that $u$ solves the variational problem in (iii). In order to recover that $u$ satisfies (ii), one first considers $v \in H^1_0(\resdom)$ and derives (using integration by parts), that $u$ satisfies the Helmholtz equation in $\resdom$, considering then a general $v \in H^1(\resdom)$ allows to conclude the equality for the boundary condition in (ii). 
\end{proof}

An immediate consequence of Proposition \ref{prop:equiv_var_non-var} is the following lemma that characterizes  the resonances of the system \eqref{eq:helmholtz_linear} variationally.
\begin{corollary}
    Let $\omega(\delta) \in \mathbb{C}$. Then, $\omega(\delta)$ is a resonance of \eqref{eq:helmholtz_linear} if and only if there exists a non-zero $u \in H^1(\resdom)$ such that
    \begin{equation*}
        \int_{\resdom}\nabla u \cdot \nabla \overline{v} \dvol - \frac{\omega(\delta)^2}{\cres^2}\int_{\resdom}u\overline{v}\dvol - \delta \int_{\partial \resdom} \mathcal{T}_{\resdom}^{\frac{\omega(\delta)}{\cback}}[u]\overline{v}\dsig = 0, \quad \forall v \in H^1(\resdom).
    \end{equation*} 
\end{corollary}
In other words, $\omega(\delta)$ is a resonance whenever the variational problem \eqref{eq:helmholtz_var} does not admit a unique solution.

In the following, we will restrict our study to the subwavelength regime. That is, we will consider $\delta$ in a neighborhood of $0 \in \mathbb{R}$ and $\omega$ in a neighborhood of $0 \in \mathbb{C}$. We will extend the variational framework given in \cite{feppon2024subwavelength} to the case of a disconnected resonator domain $\resdom$ and characterize the subwavelength resonances in that case.

The idea is as follows. The bilinear form 
\begin{equation*}
    (u,v) \longmapsto \int_{\resdom}\nabla u \cdot \nabla \overline{v} \dvol - \frac{\omega^2}{\cres^2}\int_{\resdom}u\overline{v}\dvol - \delta \int_{\partial \resdom} \mathcal{T}_{\resdom}^{\frac{\omega}{\cback}}[u]\overline{v}\dsig,
\end{equation*}
is degenerate for a continuity of $(\omega,\delta)$ in a neighborhood of $0 \in \mathbb{C}\times\mathbb{R}$, namely the subwavelength resonances $(\omega(\delta),\delta)$. While these characteristic points are objects of interest, it is unfortunate that the problem is in some sense \emph{indirect}. This complicates the analysis. The great advantage of the method presented in \cite{feppon2024subwavelength} is that it translates the indirect problem of finding characteristic points to a direct problem, where one solves a well-posed variational problem. The solution $u_{\omega,\delta}$ of that problem is then used to determine whether $(\omega,\delta)$ is a resonance.
Subsection \ref{subsec:var_form_lin} will treat the first step of developing the well-posed variational problem. In Subsection \ref{subsec:char_res_lin}, the solution of that problem will be computed and an asymptotic expansion of subwavelength resonances will be derived.

\subsection{Variational characterization of resonances}\label{subsec:var_form_lin}
In this section, we will develop a variational characterization of resonances of the linear problem \eqref{eq:helmholtz_linear}. To this end, we pose the following definition.
\begin{definition}
    Define the bilinear form $a_{\omega, \delta}: H^1(\resdom) \times H^1(\resdom) \rightarrow \mathbb{C}$ associated to the smooth domain $\resdom$ with well-separated connected components $B_1,\ldots,B_N$ that are supposed to have connected boundaries,
    \begin{equation}
        a_{\omega,\delta}(u,v) := \int_{\resdom}\nabla u \cdot \nabla \overline{v} \dvol + \sum_{j =1}^N\int_{B_j}u\dvol\int_{B_j}\overline{v}\dvol - \frac{\omega^2}{\cres^2}\int_{\resdom}u\overline{v}\dvol - \delta \int_{\partial \resdom} \mathcal{T}_{\resdom}^{\frac{\omega}{\cback}}[u]\overline{v}\dsig.
    \end{equation}
\end{definition}
The resonances of the linear Helmholtz problem \eqref{eq:helmholtz_linear} will be characterized in terms of the bilinear form $a_{\omega,\delta}$ in Theorem \ref{thm:var_lin}. In the proof of Theorem \ref{thm:var_lin}, the following lemma from \cite{feppon2024subwavelength} will be useful.

\begin{lemma}\label{lem:var_well-posed}
    There exists a neighborhood $U$ of $0 \in \mathbb{R}$ and a neighborhood $V$ of $0 \in \mathbb{C}$ such that for all $\delta \in U$ and for all $\omega \in V$ and for any $f \in H^1(\resdom)^*$, the dual of $H^1(\resdom)$ (with respect to the $H^1(\resdom)$-pairing), the following variational problem has a unique solution $u_f(\omega,\delta)$
    \begin{equation}\label{eq:var_prob_wellposed}
        a_{\omega,\delta}(u_f(\omega,\delta),v) = \langle f, v \rangle_{H^1(\resdom)^*, H^1(\resdom)},  \quad \forall\, v \in H^1(\resdom).
    \end{equation}
\end{lemma}

\begin{theorem}\label{thm:var_lin}
    Let $U \subset \mathbb{R}$ and $V \subset \mathbb{C}$ be open neighborhoods of $0$ as in Lemma \ref{lem:var_well-posed}.
    Let $\delta \in U$ and $\omega \in V$. For any $f \in H^1(\resdom)^*$, the variational problem:
    \begin{equation}\label{eq:var_prob_lin}
        \begin{split}
        \text{ find } u \in H^1(\resdom) &\text{ such that } \forall v \in H^1(\resdom),\\
        &\int_{\resdom} \left( \nabla u \cdot \nabla \overline{v} - \frac{\omega^2}{\cres^2}u\overline{v} \right) \dvol - \delta \int_{\partial \resdom} \mathcal{T}_{\resdom}^{\frac{\omega}{\cback}}[u]\overline{v}\dsig = \langle f, v \rangle_{H^1(\resdom)^*, H^1(\resdom)}
    \end{split}
    \end{equation}
    admits a unique solution $u$ if and only if 
        \begin{align}\label{eq:char_resonance}
      \ds      \Id - \left( \int_{B_i}u_{1,j}\dvol\right)_{i,j}
        \end{align}
        is invertible, where $u_{1,j} \in H^1(\resdom)$ is the unique solution to the variational problem
        \begin{align}\label{eq:var_u_1j}
            a_{\omega,\delta}(u_{1,j}(\omega,\delta),v) = \langle 1_{B_j}, v \rangle_{H^1(\resdom)^*, H^1(\resdom)} = \int_{B_j} \overline{v} \dvol, \quad \forall \,  v \in H^1(\resdom).
        \end{align}
\end{theorem}
\begin{proof}
    Clearly, \eqref{eq:var_prob_lin} is equivalent to 
    \begin{align*}
        a_{\omega,\delta}(u,v) - \sum_{j = 1}^N \int_{B_j} u \dvol\int_{B_j} \overline{v} \dvol = \langle f, v \rangle_{H^1(\resdom)^*, H^1(\resdom)}\\
        \iff a_{\omega,\delta}(u,v) = a_{\omega, \delta}(u_f(\omega,\delta),v) + \sum_{j = 1}^N \left(\int_{B_j} u \dvol\right) a_{\omega, \delta}(u_{1,j}(\omega,\delta),v),
    \end{align*}
    which implies
    \begin{align*}
        u = u_f(\omega,\delta) + \sum_{j = 1}^N\left(\int_{B_j} u \dvol\right) u_{1,j}(\omega,\delta).
    \end{align*}
    Let $\delta_{l,j}$ denote the Kronecker delta. Integrating the above equation on $B_l$ gives the following linear system of equations:
    \begin{align*}
        \begin{pmatrix}
        \ds    \int_{B_l}u_f(\omega,\delta)\dvol
        \end{pmatrix}_l = \sum_{j = 1}^N\left(\delta_{l,j} - \begin{pmatrix}
            \int_{B_l}u_{1,j}(\omega,\delta)\dvol
        \end{pmatrix}_{l,j}\right)\begin{pmatrix}
            \int_{B_j}u \dvol
        \end{pmatrix}_j, \quad l=1,\ldots, N, 
    \end{align*}
    which has a unique solution if and only if the matrix
    $$ \ds \Id - \begin{pmatrix}
        \int_{B_l}u_{1,j}(\omega,\delta)\dvol
    \end{pmatrix}_{l,j} $$
    is invertible. 
    
\end{proof}

This theorem allows us to characterize the resonances as precisely those $\omega(\delta) \in \mathbb{C}$ such that the matrix in \eqref{eq:char_resonance} is degenerate. We have the following corollary.

\begin{corollary}[Characterization of resonances]\label{cor:char_res}
    Let $\omega \in \mathbb{C}$ and $\delta \in \mathbb{R}$ and suppose that $\omega$ and $\delta$ are sufficiently small. Then, $\omega$ is a resonance of the system \eqref{eq:helmholtz_linear} if and only if \eqref{eq:char_resonance} is degenerate.
\end{corollary}

The proof of Theorem \ref{thm:var_lin} says even more. In the case of a resonance, the kernel of the matrix \eqref{eq:char_resonance} gives precisely the average in each connected component $B_j$ of the resonant solutions $u$.

\subsection{Asymptotics of linear subwavelength resonances}\label{subsec:char_res_lin}

In this subsection, the solution of the variational problem \eqref{eq:var_u_1j} will be determined for small $\omega \in \mathbb{C}$, $\delta \in \mathbb{R}$ and then, by using Corollary \ref{cor:char_res}, the asymptotics of the subwavelength resonances will be computed as $\delta \rightarrow 0$.
To this end, we consider the PDE formulation of the variational problem \eqref{eq:var_u_1j}
\begin{equation}\label{eq:PDE_u_1j}
    \left\{ 
    \begin{array}{rl}
  \ds  -\Delta u_{1,j} + \sum_{l = 1}^N\left( \int_{B_l} u_{1,j}\dvol \right) 1_{B_l} &=~~ \ds \frac{\omega^2}{\cres^2} u_{1,j} + 1_{B_j},\\
  \nm \ds
    \frac{\partial u_{1,j}}{\partial \nu} &=~~ \delta\mathcal{T}^{\frac{\omega}{\cback}}[u_{1,j}].
    \end{array}\right.
\end{equation}

We will furthermore need the following proposition.
\begin{prop}\label{prop:T1_of_chiD}
    The following equality holds:
    \begin{equation*}
        \mathcal{T}_1[1_{B_j}] = -\frac{i}{4\pi}\sum_{l = 1}^N \Capa(\resdom)_{l,j}\sum_{l=1}^N \frac{\partial\Phi_l}{\partial \nu}.
    \end{equation*}
\end{prop}

\begin{proof}
    Using item (iii) in  Lemma \ref{lem:DtN} together with Definition \ref{def:Capa} and the fact from \cite[Proposition 12]{feppon2022modal} that
    \begin{align*}
     \mathcal{S}_1[\phi] = -\frac{i}{4\pi} \int_{\partial \resdom} \phi \dsig,    
    \end{align*}
    we obtain that
    \begin{align*}
        \mathcal{T}_1[1_{B_j}] = -\mathcal{T}^0_{\resdom}\circ\mathcal{S}_{1}((\mathcal{S}^0_{\resdom})^{-1}[1_{B_j}])
        = \frac{i}{4\pi}\mathcal{T}^0_{\resdom}[1_{\resdom}]\int_{\partial \resdom} \frac{\partial \Phi_j}{\partial \nu} \dsig
        = -\sum_{l=1}^N \frac{\partial\Phi_l}{\partial \nu}\sum_{l = 1}^N \Capa(\resdom)_{l,j}.
    \end{align*}
\end{proof}

With Proposition \ref{prop:T1_of_chiD}, it is now possible to compute $u_{1,j}(\omega,\delta)$ from Theorem \ref{thm:var_lin} asymptotically as $\omega, \delta \rightarrow 0$.
\begin{prop}\label{prop:u1j}
    The solution $u_{1,j}(\omega,\delta)$ to \eqref{eq:var_u_1j} has the following asymptotic behavior as $\omega, \delta \rightarrow 0$:
    \begin{align*}
        u_{1,j}(\omega,\delta) = 
        \frac{1}{\lvert B_j \rvert} 1_{B_j} +\omega^2 \frac{1}{\cres^2\lvert B_j \rvert^2} 1_{B_j}
        +\delta \left(-\sum_{l=1}^N\frac{\Capa(\resdom)_{l,j}}{\lvert B_j \rvert\lvert B_l \rvert^2} 1_{B_l} + \tilde{u}^j_{0,1}\right)\\
        +\omega\delta \left(\sum_{l = 1}^N \left(\frac{i}{4\pi\lvert B_j \rvert\lvert B_l \rvert^2\cback} (\Capa(\resdom)J\Capa(\resdom))_{l,j}\right) 1_{B_l} + \tilde{u}^j_{1,1}\right) + O((\omega^2 + \delta)^2),
    \end{align*}
    where $J$ denotes the $(N\times N)$-matrix of ones.
\end{prop}

\begin{proof}
We make the Ansatz 
$$ u_{1,j}(\omega,\delta) = \sum_{p,k \geq 0} \omega^p \delta^k u_{p,k}, $$
inserting this into \eqref{eq:PDE_u_1j} and identifying powers of $\delta$ and $\omega$, we obtain the following cascade of equations characterizing the functions $(u_{p,k})_{p,k\geq 0} $: 
\begin{align*}
    -\Delta u_{p,k} + \sum_{i = 1}^N\left( \int_{B_i} u_{p,k}\dvol \right) 1_{B_i} &= \frac{1}{\cres^2} u_{p-2,k} + 1_{B_j}\delta_{p=0}\delta_{k=0},\\
    \frac{\partial u_{p,k}}{\partial \nu} &= \sum_{n = 0}^p \frac{1}{\cback^n}\mathcal{T}_n[u_{p-n,k-1}].  
\end{align*}
First, we consider $ p \geq 0$ and $k = 0$
\begin{align*}
    -\Delta u_{p,0} + \sum_{i = 1}^N\left( \int_{B_i} u_{p,0}\dvol \right) 1_{B_i} &= \frac{1}{\cres^2} u_{p-2,0} + 1_{B_j}\delta_{p=0},\\
    \frac{\partial u_{p,0}}{\partial \nu} &= 0.
\end{align*}
One can prove inductively that
\begin{align*}
    u_{2p,0} = \frac{1}{c_r^{2p}\lvert B_j \rvert^{p+1}} 1_{B_j}, \quad u_{2p + 1, 0} = 0, \quad \text{ for } \quad p \geq 0.
\end{align*}
Let $p = 0, k = 1$. Then $u_{0,1}$ satisfies the following system of equations:
\begin{align*}
    -\Delta u_{0,1} + \sum_{i = 1}^N\left( \int_{B_i} u_{0,1}\dvol \right) 1_{B_i} &= 0,\\
    \frac{\partial u_{0,1}}{\partial \nu} &= \mathcal{T}_0[u_{0,0}]. 
\end{align*}
The integral of $u_{0,1}$ on $B_l$ is thus given by
\begin{equation*}
    \lvert B_l \rvert \int_{B_l} u_{0,1}\dvol = \int_{\partial B_l} \frac{\partial u_{0,1}}{\partial \nu}\dsig  = \int_{\partial B_l} \mathcal{T}^0[u_{0,0}]\dsig = \frac{1}{\lvert B_j \rvert}\int_{\partial B_l} \frac{\partial \Phi_j}{\partial \nu}\dsig
    = -\frac{1}{\lvert B_j \rvert}\Capa(\resdom)_{l,j}.
\end{equation*}
It follows that $u_{0,1}$ is given by 
\begin{equation*}
    u_{0,1} = - \sum_{l=1}^N \frac{1}{\lvert B_j \rvert\lvert B_l \rvert^2}\Capa(\resdom)_{l,j} 1_{B_l} + \tilde{u}_{0,1}^j,
\end{equation*}
where $\tilde{u}_{0,1}^j$ satisfies
\begin{equation*}\left\{
    \begin{array}{rll}
        \ds \Delta \tilde{u}_{0,1}^j &=~~ \ds- \sum_{l=1}^N \frac{1}{\lvert B_j \rvert\lvert B_l \rvert}\Capa(\resdom)_{l,j} 1_{B_l} & \text{in } \resdom,\\
        \nm \ds
    \frac{\partial \tilde{u}_{0,1}^j}{\partial \nu} &=~~ \ds \frac{1}{\lvert B_j \rvert}\frac{\partial \Phi_j}{\partial \nu} &\text{on } \partial \resdom,\\
    \nm \ds
    \int_{B_l} \tilde{u}_{0,1}^j \dvol &=~~ 0, & \forall ~l \in \{1,\ldots,N\}.
    \end{array}\right.
\end{equation*}
For $ p, k = 1,$ one obtains  that
\begin{align*}
    -\Delta u_{1,1} + \sum_{i = 1}^N\left( \int_{B_i} u_{1,1}\dvol \right) 1_{B_i} &= 0,\\
    \frac{\partial u_{1,1}}{\partial \nu} &= \mathcal{T}_0[u_{1,0}] + \frac{1}{\cback}\mathcal{T}_1[u_{0,0}] = \frac{1}{\cback}\mathcal{T}_1[u_{0,0}].
\end{align*}
Using Proposition \ref{prop:T1_of_chiD} yields
\begin{align*}
    \frac{1}{\cback}\mathcal{T}_1[u_{0,0}] = \frac{1}{\cback\lvert B_j \rvert}\mathcal{T}_1[1_{B_j}] &= -\frac{i}{4 \pi \cback\lvert B_j \rvert} \sum_{l = 1}^N \Capa(\resdom)_{l,j}\sum_{l=1}^N \frac{\partial\Phi_l}{\partial \nu}.
\end{align*}
Therefore, it follows that
\begin{align*}
    \lvert B_m \rvert \int_{B_m} u_{1,1}\dvol = \frac{-i}{4 \pi \cback\lvert B_j \rvert} \sum_{l = 1}^N\Capa(\resdom)_{l,j}\int_{\partial B_m} \sum_{l = 1}^N\frac{\partial \Phi_l}{\partial \nu} \dsig\\
    = \frac{i}{4 \pi \cback\lvert B_j \rvert} \sum_{l = 1}^N\Capa(\resdom)_{l,j} \sum_{l = 1}^N \Capa(\resdom)_{m,l}\\
    = \frac{i}{4 \pi \cback\lvert B_j \rvert} (\Capa(\resdom) J \Capa(\resdom))_{m,j},
\end{align*}
and we can decompose $u_{1,1}$ as 
\begin{align*}
    u_{1,1} = \sum_{m=1}^N \frac{i}{4 \pi \cback\lvert B_j \rvert\lvert B_m \rvert^2} (\Capa(\resdom) J \Capa(\resdom))_{m,j} 1_{B_m} + \tilde{u}_{1,1}^j,
\end{align*}
where $\tilde{u}_{1,1}^j$ is the solution to
\begin{equation*}\left\{
    \begin{array}{rll}
        \Delta\tilde{u}_{1,1}^j &=~~ \ds \sum_{m=1}^N \frac{i}{4 \pi \cback\lvert B_j \rvert\lvert B_m \rvert} (\Capa(\resdom) J \Capa(\resdom))_{m,j} 1_{B_m} &\text{in } \resdom,\\
        \nm \ds
        \frac{\partial \tilde{u}_{1,1}^j}{\partial \nu} &=~~ \ds -\frac{i}{4 \pi \cback\lvert B_j \rvert} \sum_{l = 1}^N\Capa(\resdom)_{l,j} \sum_{l = 1}^N\frac{\partial \Phi_l}{\partial \nu} &\text{on } \partial \resdom,\\
        \nm \ds
        \int_{B_l} \tilde{u}_{1,1}^j \dvol &=~~ 0, &\forall ~l \in \{1,\ldots,N\}.
    \end{array}\right.
\end{equation*}

\end{proof}

Applying Corollary \ref{cor:char_res} one can reproduce the fundamental theorem of subwavelength physics, Theorem \ref{thm:fundamental}.

\begin{theorem}[Fundamental theorem of subwavelength physics]
    Let $\{q_0,v_2,\ldots,v_N\}$ be a basis of eigenvectors of $\Capa^{gen}(\resdom)$ and denote by $\Pi: \mathbb{C}^{N} \rightarrow \mathbb{C}q_0$ the projection along $\Span(v_2,\ldots,v_N)$ onto $\mathbb{C}q_0$.
    Then the subwavelength resonance $\omega(\delta) = \omega_0\sqrt{\delta} + \omega_1\delta + O(\delta^{\frac{3}{2}})$ of \eqref{eq:helmholtz_linear} associated to $q_0$ satisfies the following asymptotics:
    \begin{align*}
        \omega_0 \text{ is an eigenvalue of } \Capa^{gen}(\resdom) \text{ with eigenvector } q_0,\\
        \omega_1 = -\frac{i}{8\pi\cback} \frac{q_0 \cdot \Pi\left[  \Capa^{gen}(\resdom)J\Capa(\resdom)q_0 \right]}{\lVert q_0 \rVert^2}.
    \end{align*}
\end{theorem}

\begin{proof}
    By Theorem \ref{thm:var_lin}, the subwavelength resonances of \eqref{eq:helmholtz_linear} are given by those $\omega(\delta)$ such that \eqref{eq:char_resonance} is not invertible. Inserting the results of Proposition \ref{prop:u1j} into \eqref{eq:char_resonance} implies that $\omega(\delta)$ is a subwavelength resonance whenever there exists $q(\delta) \in \mathbb{C}^N\setminus\{0\}$ such that
    \begin{align*} 
        \left(\omega(\delta)^2 - \delta \Capa^{gen}(\resdom) - \omega(\delta)\delta \frac{i}{4\pi\cback} \Capa^{gen}(\resdom)J\Capa(\resdom)\right)q(\delta) = 0.
    \end{align*}
    Making the Ansatz $\omega(\delta) = \omega_0 \sqrt{\delta} + \omega_1 \delta + O(\delta^{\frac{3}{2}})$ and $q(\delta) = q_0 + q_1\sqrt{\delta} + O(\delta)$ and inserting it into the above equation leads to the following cascade of equations:
    \begin{align*}
        \left(\omega_0^2 - \Capa^{gen}(\resdom)\right)q_0 = 0,\\
        \left(\omega_0^2 - \Capa^{gen}(\resdom)\right)q_1 +  \omega_0\left(2\omega_1q_0 - \frac{i}{4\pi\cback} \Capa^{gen}(\resdom)J\Capa(\resdom)q_0\right) = 0.
    \end{align*}
    It then follows that $\omega_0$ is an eigenvalue of $\Capa^{gen}(\resdom)$ with an eigenvector $q_0$ and $\omega_1$ is given by 
    \begin{align*}
        \omega_1 = \frac{q_0 \cdot \Pi\left[ \frac{i}{8\pi\cback} \Capa^{gen}(\resdom)J\Capa(\resdom)q_0 \right]}{\lVert q_0 \rVert^2}.
    \end{align*}
\end{proof}

In the next section, a fundamental theorem of \emph{nonlinear} subwavelength physics will be derived.

\section{Nonlinear subwavelength resonances}\label{section:non-linear}
In this section, similar techniques to those in Section \ref{section:linear} are used to derive a fundamental theorem of \emph{nonlinear} subwavelength physics, Theorem \ref{thm:fundamental_non-linear}. The main difference is that the bilinear form $a_{\omega,\delta}$ in Theorem \ref{thm:var_lin} will not be bilinear in this section, but will account for the nonlinearity $g(u)$ in \eqref{eq:wave_eq_non-linear}. Consequently, in a formulation similar to \eqref{eq:char_resonance}, it will be important to account for the different amplitudes and the nonlinear mixing of the different $u_{1,j}$. Concretely, the condition (\ref{eq:char_resonance}-\ref{eq:var_u_1j}) in Theorem \ref{thm:var_lin} becomes a nonlinear system (\ref{eq:u_f_nonlin}-\ref{eq:char_resonance_nonlin}) in one function $u_f(\omega,\delta) \in H^1(\resdom)$ in Theorem \ref{thm:var_nonlin}.

The remaining steps of computing $u_f(\omega,\delta)$ asymptotically and characterizing the subwavelength resonances are then identical to the linear setting. The asymptotics of $u_f(\omega,\delta)$ can be found in Proposition \ref{prop:u_f} and the fundamental theorem of nonlinear subwavelength physics is given in Theorem \ref{thm:fundamental_non-linear}.

It is important to note that the techniques in this section are not limited to $g(u) = \lvert \frac{\partial u}{\partial t} \rvert^2 \frac{\partial u}{\partial t}$, but can be used for any $g(u)$ that satisfies that $g(\exp(i \omega t)u(x))$ for $u \in H^1(\resdom)$ is of order $O(\omega^2)$ as $\omega \rightarrow 0$ and that satisfies a gauge invariance in the sense that 
\begin{equation*}
    g(\exp(i \omega t)u(x)) = \exp(i\omega t) v(x),
\end{equation*}
for some $v \in H^1(\resdom)$ that depends on $u$ and $\omega$ but is constant with respect to $t$.

Similar to Proposition \ref{prop:equiv_var_non-var}, the system \eqref{eq:helmholtz_non-linear} can equivalently be formulated with the Dirichlet-to-Neumann operator as follows:
\begin{equation*}
    \left\{
    \begin{array}{rll}
      \ds   \Delta u + \frac{\omega^2}{\cres^2} u &\hspace{-5pt}= 
      {-}i\beta\lvert \omega\rvert^2\omega \lvert u\rvert^2u &\text{in } \resdom,\\
        \nm \ds
        \frac{\partial u}{\partial \nu} &\hspace{-5pt}= \delta\mathcal{T}_{\resdom}^{\frac{\omega}{\cback}}[u] &\text{on } \partial\resdom,
    \end{array}
    \right.
\end{equation*}
with $u: \resdom \rightarrow \mathbb{C}$.

This is equivalent to the following variational problem: Find $u \in H^1(\resdom)$ such that 
    \begin{equation*}\label{eq:helmholtz_nonlin_var}
        \begin{split}
            \int_{\resdom}\nabla u \cdot \nabla \overline{v} \dvol - \frac{\omega^2}{\cres^2}\int_{\resdom}u\overline{v}\dvol -\int_{\resdom}i\beta\lvert \omega\rvert^2\omega \lvert u\rvert^2u\overline{v}\dvol - \delta \int_{\partial \resdom} \mathcal{T}_{\resdom}^{\frac{\omega}{\cback}}[u]\overline{v}\dsig = 0,  \quad\quad \\
            \quad \forall v \in H^1(\resdom).
        \end{split}
    \end{equation*}
A resonance $\omega(\delta)$ is given whenever there exists $u_{\omega,\delta} \in H^1(\partial \resdom)\setminus\{0\}$ solving this variational problem. 

\begin{definition}[Nonlinear subwavelength resonance]
    Let $\omega: (0,\epsilon) \rightarrow \mathbb{C}$ be a continuous function for some $\epsilon>0$. We say that $\omega(\delta)$ is a \emph{resonance} of the system \eqref{eq:helmholtz_non-linear}, if for every $\delta \in (0,\epsilon)$ the variational problem \eqref{eq:helmholtz_nonlin_var} has a non-zero solution $u_{\omega,\delta} \in H^1(\resdom)$. \\
    A resonance $\omega(\delta)$ is said to be \emph{subwavelength} if it also satisfies $ \omega(\delta) \rightarrow 0 \text{ as } \delta \rightarrow 0$.
\end{definition}

\subsection{Variational characterization of nonlinear resonances}\label{subsec:var_form_non-lin}
The following \emph{semilinear} form will play the role of $a_{\omega,\delta}$ in the nonlinear setting.
\begin{definition}
    Suppose that the smooth, bounded domain $\resdom \in \R3$ has well-separated connected components $B_1,\ldots,B_N$ with connected boundaries. 
    Define the semilinear form $\tilde{a}_{\omega, \delta}: H^1(\resdom) \times H^1(\resdom) \rightarrow \mathbb{C}$ associated to $\resdom$ by
    \begin{equation}
        \begin{split}
            \tilde{a}_{\omega,\delta}(u,v) := \int_{\resdom}\nabla u \cdot \nabla \overline{v} \dvol + \sum_{j =1}^N\int_{B_j}u\dvol\int_{B_j}\overline{v}\dvol - \frac{\omega^2}{\cres^2}\int_{\resdom}u\overline{v}\dvol \quad\quad\\ - \int_{\resdom}i\beta\lvert \omega\rvert^2\omega \lvert u\rvert^2u\overline{v}\dvol - \delta \int_{\partial \resdom} \mathcal{T}_{\resdom}^{\frac{\omega}{\cback}}[u]\overline{v}\dsig.
        \end{split}
    \end{equation}
\end{definition}

Recall here from the Gagliardo-Nirenberg-Sobolev inequalities that $H^1(\resdom) \subset L^4(\resdom)$. With the form $\tilde{a}_{\omega,\delta}$ it is possible to characterize the resonances of the nonlinear system \eqref{eq:helmholtz_non-linear} in a similar fashion as in the linear case in Theorem \ref{thm:var_lin} and Corollary \ref{cor:char_res}. To this end, let $\beta$ be fixed. For $f \in L^\infty(\resdom)$, we can obtain from \cite{jun} (see also \cite{lutz,roland}) that there exists a neighborhood $U$ of $0 \in \mathbb{R}$ and a neighborhood $V$ of $0 \in \mathbb{C}$ such that for all $\delta \in U$ and for all $\omega \in V$, the following variational problem:
\begin{equation*}\label{eq:var_prob_wellposed_nonlin}
    \tilde{a}_{\omega,\delta}(u_f(\omega,\delta),v) = \langle f, v \rangle_{H^1(\resdom)^*, H^1(\resdom)}, \quad \forall v \in H^1(\resdom),
\end{equation*}
has a unique solution $u_f(\omega,\delta)$. 

With this property at hand, the following variational characterization of small resonances of the nonlinear system \eqref{eq:helmholtz_non-linear} is possible.

\begin{theorem}[Characterization of resonances]\label{thm:var_nonlin}
    Let $\omega \in \mathbb{C}$ and $\delta \in \mathbb{R}$ belong to sufficiently small neighborhoods of zero. Then $\omega$ is a resonance of the system \eqref{eq:helmholtz_non-linear}, if and only if there  exist $a_1,\ldots,a_N \in \mathbb{C}$ not all zero such that the solution $u_f$ to
    \begin{align}\label{eq:u_f_nonlin}
        \tilde{a}_{\omega,\delta}(u_f, v) = \sum_{i=1}^N a_i \int_{B_i}\overline{v}\dvol,  \quad \forall v \in H^1(\resdom),
    \end{align}
   \emph{i.e.}, with $\ds f= \sum_{i=1}^N a_i \int_{B_i} 1_{B_j}$, satisfies
        \begin{align}\label{eq:char_resonance_nonlin}
            \int_{B_i} u_f \dvol = a_i, \quad \forall i \in \{1, \ldots, N\}.
        \end{align}
\end{theorem}

\begin{proof}
    Assume that $\omega$ is a resonance of the system \eqref{eq:helmholtz_non-linear}. Then, there exists $u_f \in H^1(\resdom)\setminus \{0\}$ such that $u_f$ satisfies \eqref{eq:helmholtz_nonlin_var}.
    Set $a_i := \int_{B_i} u_f \dvol$. Then, $u_f$ satisfies
    \begin{align*}
        \tilde{a}_{\omega,\delta}(u_f, v) = \sum_{i=1}^N \int_{B_i} u_f \dvol\int_{B_i}\overline{v}\dvol = \sum_{i=1}^N a_i \int_{B_i}\overline{v}\dvol ,  \quad \forall v \in H^1(\resdom),
    \end{align*}
    where the $a_i$ are not all zero.
    Conversely, assume there exist $a_1,\ldots, a_N \in \mathbb{C}$ not all zero and $u_f \in H^1(\resdom)$ as in the theorem. Let $v \in H^1(\resdom)$, then using \eqref{eq:u_f_nonlin} and \eqref{eq:char_resonance_nonlin} one obtains
    \begin{align*}
        \int_{\resdom}\nabla u_f \cdot \nabla \overline{v} \dvol - \frac{\omega^2}{\cres^2}\int_{\resdom}u_f\overline{v}\dvol -\int_{\resdom}i\beta\lvert \omega\rvert^2\omega \lvert u_f\rvert^2u_f\overline{v}\dvol - \delta \int_{\partial \resdom} \mathcal{T}_{\resdom}^{\frac{\omega}{\cback}}[u_f]\overline{v}\dsig \\
        = \tilde{a}_{\omega,\delta}(u_f, v) - \sum_{i=1}^N \int_{B_i} u_f \dvol\int_{B_i}\overline{v}\dvol\\
        = 0.
    \end{align*}
    Hence, $u_f$ satisfies the variational problem \eqref{eq:helmholtz_nonlin_var}, is not identically zero and thus $\omega$ is a resonance.
\end{proof}

\subsection{Asymptotics of nonlinear subwavelength resonances}\label{subsec:char_res_non-lin}

In this subsection, we will determine the solution of the variational problem \eqref{eq:u_f_nonlin} asymptotically and then derive asymptotics for subwavelength resonances of the nonlinear system \eqref{eq:helmholtz_non-linear}. To this end, we consider the PDE formulation of the variational problem \eqref{eq:u_f_nonlin}
\begin{equation}\label{eq:PDE_u_f}
    \left\{ 
    \begin{array}{rl}
  \ds  -\Delta u_f + \sum_{j = 1}^N a_j\left( \int_{B_j} u_f\dvol \right)1_{B_j} &=~~ \ds \frac{\omega^2}{\cres^2} u_f + i\beta\lvert\omega\rvert^2\omega\lvert u_f \rvert^2 u_f+ \sum_{j =1}^Na_j 1_{B_j},\\
    \nm \ds
    \frac{\partial u_f}{\partial \nu} &=~~ \delta\mathcal{T}^{\frac{\omega}{\cback}}[u_f].
    \end{array}\right.
\end{equation}

\begin{prop}\label{prop:u_f}
    The solution $u_f(\omega,\delta)$ to \eqref{eq:u_f_nonlin} has the following asymptotic behavior as $\omega, \delta \rightarrow 0$:
    \begin{align*}
        u_f(\omega,\delta) = \sum_{j=1}^N
        \frac{a_j}{\lvert B_j \rvert} 1_{B_j} +\omega^2 \sum_{j=1}^N\frac{a_j}{\cres^2\lvert B_j \rvert^2}1_{B_j}
        +\delta \sum_{j=1}^Na_j\left(-\sum_{l=1}^N\frac{\Capa(\resdom)_{l,j}}{\lvert B_j \rvert\lvert B_l \rvert^2} 1_{B_l} + \tilde{u}^j_{0,1}\right)\\
        +\omega\delta \sum_{j=1}^Na_j\left(\sum_{l = 1}^N \left(\frac{i}{4\pi\lvert B_j \rvert\lvert B_l \rvert^2\cback} (\Capa(\resdom)J\Capa(\resdom))_{l,j}\right)1_{B_l} + \tilde{u}^j_{1,1}\right)\\
        + \lvert \omega \rvert^2 \omega i\beta \sum_{j=1}^N \frac{\lvert a_j \rvert^2 a_j}{\lvert B_j \rvert^4} 1_{B_j} + O((\omega^2 + \delta)^2),
    \end{align*}
    where $J$ denotes the $(N\times N)$-matrix of ones.
\end{prop}

\begin{proof}
    We make the Ansatz 
$$ u_f(\omega,\delta) = \sum_{p,k,l \geq 0} \omega^p \delta^k \overline{\omega}^l u_{p,k,l}.$$
Inserting this into \eqref{eq:PDE_u_f} and identifying powers of $\delta$ and $\omega$, we obtain the following cascade of equations characterizing the functions $(u_{p,k,l})_{p,k\geq 0}$: 
\begin{align*}
    -\Delta u_{p,k,l} + \sum_{i = 1}^N\left( \int_{B_i} u_{p,k,l}\dvol \right)1_{B_i} &= \frac{1}{\cres^2} u_{p-2,k,l} + i\beta\left( \lvert u_f \rvert^2 u_f \right)_{p-2,k,l-1} + \sum_{j =1}^Na_j1_{B_j}\delta_{p=0}\delta_{k=0}\delta_{l=0},\\
    \frac{\partial u_{p,k,l}}{\partial \nu} &= \sum_{n = 0}^p \frac{1}{\cback^n}\mathcal{T}_n[u_{p-n,k-1,l}],
\end{align*}
where $\left( \lvert u_f \rvert^2 u_f \right)_{p-2,k,l-1}$ is the $\omega^{p-2}\delta^k\overline{\omega}^{l-1}$-component  of $\lvert u_f \rvert^2 u_f$.
For $l=0$, the same computations as in the proof of Proposition \ref{prop:u1j} hold and one obtains that
\begin{align*}
    u_{p,k,0} = \sum_{j = 1}^N a_j(u_{1,j})_{p,k}.
\end{align*}
Assume $l=1$. Then $u_{0,0,1} = u_{1,0,1} = 0$. For $l=1,p=2,k=0,$ one has
\begin{align*}
    -\Delta u_{2,0,1} + \sum_{i = 1}^N\left( \int_{B_i} u_{2,0,1}\dvol \right)1_{B_i} &=i\beta\left( \lvert u_f \rvert^2 u_f \right)_{0,0,0},\\
    \frac{\partial u_{2,0,1}}{\partial \nu} &= 0.
\end{align*}
Thus, 
\begin{align*}
    \lvert B_j \rvert \int_{B_j} u_{2,0,1}(x)\dvol = i\beta \int_{B_j} \lvert u_{0,0,0}\rvert^2u_{0,0,0}\dvol , 
\end{align*}
and it follows that
\begin{align*}
    u_{2,0,1}(x) = i\beta \frac{\lvert u_{0,0,0}(x)\rvert^2u_{0,0,0}(x)}{\lvert B_j \rvert} \quad \text{ for } x \in B_j.
\end{align*}
Similarly, it follows that $u_{0,1,1} = u_{1,0,2} = u_{0,0,3} = 0$.
\end{proof}

This proposition allows us to derive the following fundamental theorem of nonlinear subwavelength physics by deriving the asymptotics of the nonlinear subwavelength resonances in terms of $\delta$.

\begin{theorem}\label{thm:fundamental_non-linear}
    Let $\{q_0,v_2,\ldots,v_N\}$ be a basis of eigenvectors of $\Capa^{gen}(\resdom)$, let $\V$ be as in Definition \ref{def:Capa} and denote by $\Pi: \mathbb{C}^{N} \rightarrow \mathbb{C}q_0$ the projection along $\Span(v_2,\ldots,v_N)$ onto $\mathbb{C}q_0$.
    Then the subwavelength resonance $\omega(\delta) = \omega_0\sqrt{\delta} + \omega_1\delta + O(\delta^{\frac{3}{2}})$ of \eqref{eq:helmholtz_non-linear} associated to $q_0$  satisfies the following asymptotics
    \begin{align*}
        \omega_0 \text{ is an eigenvalue of } \Capa^{gen}(\resdom) \text{ with eigenvector } q_0,\end{align*}
        {\begin{equation} \label{for:same1}
        \omega_1 = \frac{q_0 \cdot \Pi\left[ \frac{-i}{4\pi\cback} \Capa^{gen}(\resdom)J\Capa(\resdom)q_0 - \lvert \omega_0 \rvert^2 i\beta \cres^2 \lvert q_0 \rvert^2 q_0\right]}{2\lVert q_0 \rVert^2},
    \end{equation}
    }
    where $\lvert q_0 \rvert^2 q_0$ is to be understood component-wise.
\end{theorem}

\begin{proof}
    By Theorem \ref{thm:var_nonlin}, the subwavelength resonances of \eqref{eq:helmholtz_non-linear} are given by those $\omega(\delta)$ such that there exist $a_1,\ldots,a_N$ such that \eqref{eq:u_f_nonlin} and \eqref{eq:char_resonance_nonlin} are satisfied. Inserting the results of Proposition \ref{prop:u_f} into \eqref{eq:char_resonance_nonlin} implies that $\omega(\delta)$ is a subwavelength resonance whenever there exists $q(\delta) \in \mathbb{C}^N\setminus\{0\}$ such that
    \begin{align} \label{cap:NL}
    \begin{split}
        \left(\omega(\delta)^2 - \delta \Capa^{gen}(\resdom) + \omega(\delta)\delta \frac{i}{4\pi\cback} \Capa^{gen}(\resdom)J\Capa(\resdom)\right)q(\delta) \quad\quad\quad\quad\quad\quad\\
         + \lvert \omega(\delta) \rvert^2 \omega(\delta) i\beta  \cres^2 \lvert q(\delta) \rvert^2 q(\delta) = 0,
    \end{split}
    \end{align}
    where $ \lvert q(\delta) \rvert^2 q(\delta) $ is to be understood component-wise and $q_j = \frac{a_j}{\lvert B_j \rvert}$.
    Making the Ansatz $\omega(\delta) = \omega_0 \sqrt{\delta} + \omega_1 \delta + O(\delta^{\frac{3}{2}})$ and $q(\delta) = q_0 + q_1\sqrt{\delta} + O(\delta)$ and inserting it into the above equation leads to the following cascade of equations:
    \begin{align*}
        \left(\omega_0^2 - \Capa^{gen}(\resdom)\right)q_0 = 0,\\
        \left(\omega_0^2 - \Capa^{gen}(\resdom)\right)q_1 +  \omega_0\left(2\omega_1q_0 + \frac{i}{4\pi\cback} \Capa^{gen}(\resdom)J\Capa(\resdom)q_0 + \lvert \omega_0 \rvert^2 i\beta  \cres^2 \lvert q_0 \rvert^2 q_0 \right)= 0.
    \end{align*}
    It then follows that $\omega_0$ is an eigenvalue of $\Capa^{gen}(\resdom)$ with eigenvector $q_0$ and $\omega_1$ is given by 
    \begin{align*}
        \omega_1 = \frac{q_0 \cdot \Pi\left[ \frac{-i}{\pi\cback} \Capa^{gen}(\resdom)J\Capa(\resdom)q_0 - \lvert \omega_0 \rvert^2 i\beta  \cres^2 \lvert q_0 \rvert^2 q_0\right]}{2\lVert q_0 \rVert^2}.
    \end{align*}
\end{proof}

\begin{remark}
    Note that \eqref{cap:NL} has some similarities with the discrete Gross-Pitaevskii and Aubry-Andr\'e-Harper models, which are used as discrete approximations of 
    nonlinear Schr\"odinger equations (with cubic nonlinearities) in quantum mechanics; see, \emph{e.g.}, \cite{soliton3,soliton4,soliton6}. Such quantum mechanical models are typically based on tight-binding approximations and include only nearest-neighbour interaction terms (long-range interactions are neglected). They correspond to \eqref{cap:NL} with $\Capa^{gen}$ being tridiagonal \cite{flach}. We refer to \cite{weinstein1,weinstein2,weinstein3,localization,pankov} for their analysis. 
\end{remark}

\begin{remark}
    As in \cite{cochlea}, if we consider the scattering of an incident wave $u_{in}$ by a system of nonlinear subwavelength resonators, then based on Theorem \ref{thm:fundamental_non-linear}, we can approximate the scattered field $u-u_{in}$ by a few subwavelength modes involving only the nonlinear subwavelength resonances and their associated eigenmodes, so as to retain the physically derived coupling between resonators.  
\end{remark}

As already mentioned at the beginning of this section, the above procedure is not limited to $g(u) = \lvert \frac{\partial u}{\partial t} \rvert^2\frac{\partial u}{\partial t}$ but solely relies on the property that $g(u)$ is gauge invariant and $g(u(x)\exp(i\omega t))$ is of order $O(\omega^2)$ as $\omega \rightarrow 0$. The following remark is an illustration of this fact if the order is $O(\omega^2)$ as $\omega \rightarrow 0$.
\begin{remark}\label{rem:different_nonlin}
    If one changes the nonlinearity $i \beta \lvert\omega\rvert^2 \omega\lvert u\rvert^2 u$ to $- \beta \omega^2\lvert u\rvert^2 u$, the nonlinearity is of order $\omega^2$ thus influencing the leading-order behavior of the subwavelength resonances.
    To be precise, the cascading system of equations from the proof of Proposition \ref{prop:u_f} is then given by 
    \begin{align*}
        -\Delta u_{p,k,l} + \sum_{i = 1}^N\left( \int_{B_i} u_{p,k,l}\dvol \right)1_{B_i} &= \frac{1}{\cres^2} u_{p-2,k,l} - \beta\left( \lvert u_f \rvert^2 u_f \right)_{p-2,k,l} + \sum_{j =1}^Na_j1_{B_j}\delta_{p=0}\delta_{k=0}\delta_{l=0},\\
        \frac{\partial u_{p,k,l}}{\partial \nu} &= \sum_{n = 0}^p \frac{1}{\cback^n}\mathcal{T}_n[u_{p-n,k-1,l}],
    \end{align*}
    and one obtains
    \begin{equation*}\begin{split}
        u_f(\omega,\delta) = \sum_{j=1}^N
        \frac{a_j}{\lvert B_j \rvert} 1_{B_j} +\omega^2 \left(\sum_{j=1}^N\frac{a_j}{\cres^2\lvert B_j \rvert^2}1_{B_j} - \beta \sum_{j=1}^N \frac{\lvert a_j \rvert^2 a_j}{\lvert B_j \rvert^4} 1_{B_j}\right)\quad\quad\quad\quad\quad\quad\\
        +\delta \sum_{j=1}^Na_j\left(-\sum_{l=1}^N\frac{\Capa(\resdom)_{l,j}}{\lvert B_j \rvert\lvert B_l \rvert^2} 1_{B_l} + \tilde{u}^j_{0,1}\right)\quad\quad\quad\quad\\
        +\omega\delta \sum_{j=1}^Na_j\left(\sum_{l = 1}^N \left(\frac{i}{4\pi\lvert B_j \rvert\lvert B_l \rvert^2\cback} (\Capa(\resdom)J\Capa(\resdom))_{l,j}\right)1_{B_l} + \tilde{u}^j_{1,1}\right)
          + O((\omega^2 + \delta)^2).
    \end{split}
    \end{equation*}
    In that case, the leading order asymptotics of the subwavelength resonances $\omega(\delta) = \omega_0\sqrt{\delta} + O(\delta)$ are characterized by the nonlinear eigenvalue problem
   {\begin{equation} \label{for:same2} \Capa^{gen}(\resdom)q_0 - \omega_0^2\left( q_0 - \beta \cres^2 \lvert q_0 \rvert^2 q_0\right) = 0, \end{equation}
   }
    where $\lvert q_0 \rvert^2 q_0$ is to be understood to act component-wise.
    If $\resdom$ is connected, then
    \begin{align}\label{eq:res_monomer}
        \omega_0^2 = \frac{\Capa^{gen}(\resdom)}{1 - \beta\cres^2\lvert q_0 \rvert^2}.
    \end{align}
\end{remark}

\section{Numerical experiments: A nonlinear system of two resonators}\label{section:num}
To illustrate the results from Section \ref{section:non-linear}, some numerical results will be presented in this section.
Concretely, we will study the nonlinearity $-\beta\omega^2\lvert u \rvert^2 u$ for the Helmholtz system \eqref{eq:helmholtz_non-linear}, as discussed in Remark \ref{rem:different_nonlin}.
The describing equation for the leading-order behavior of the subwavelength modes, is in that case given by
\begin{equation}\label{eq:char_resonance_num}
    \Capa^{gen}(\resdom)q_0 - \omega_0^2\left( q_0 - \beta \cres^2 \lvert q_0 \rvert^2 q_0\right) = 0,
\end{equation}
where the entries of the vector $q_0$ correspond to the average of the subwavelength resonant solutions inside the connected components $B_1, \ldots, B_N$ and $\lvert q_0 \rvert^2 q_0$ is to be understood to act component-wise on the entries of $q_0$.

In the case where $\resdom = B_1 \cup B_2$ and $B_1$ and $B_2$ satisfy certain symmetry conditions, the leading order behavior of the resonant solutions to the linear system \eqref{eq:helmholtz_linear} is not influenced by the nonlinearity, however the leading order behavior of the resonances changes.
\begin{lemma}\label{lem:non-lin_sol}
    Assume that $\resdom$ is a smooth bounded domain with well-separated connected components $B_1,~ B_2$ with connected boundaries.
    If $\resdom = B_1 \cup B_2$ is preserved under the mirror symmetry $ P(x_1,x_2,x_3) = (x_1,x_2,-x_3) $ and $P(B_1) = P(B_2)$, then, for all $a \in \mathbb{C}$,
    $$ v_0 = \begin{pmatrix}
        a\\
        a
    \end{pmatrix} \text{ and }~ v_1 = \begin{pmatrix}
        -a\\
        a
    \end{pmatrix}$$ 
    with 
    \begin{align}\label{eq:resonance_dimer}
        \omega_0 = \frac{\Capa^{gen}(\resdom)_{1,1} + \Capa^{gen}(\resdom)_{2,2}}{\left(1  - \beta \cres^2 \lvert a \rvert^2 \right)} \quad \text{ and } \quad \omega_1 = \frac{\Capa^{gen}(\resdom)_{1,1} - \Capa^{gen}(\resdom)_{2,2}}{\left(1  - \beta \cres^2 \lvert a \rvert^2 \right)},
    \end{align}
    respectively, are solutions to equation \eqref{eq:char_resonance_num}.
\end{lemma}
\begin{proof}
    Since $\resdom$ satisfies the mirror symmetry condition, it follows that $\Capa(\resdom)$ has the form
    $$ \Capa(\resdom) = \begin{pmatrix}
        C_1 & C_2\\ C_2 &C_1
    \end{pmatrix}, $$ with $C_1, C_2 \in \mathbb{R}$. Furthermore, $(\V)_{1,1} = (\V)_{2,2} = \lvert B_1 \rvert^{-1} = \lvert B_2 \rvert^{-1}$. Thus, a basis of eigenvectors of $\Capa^{gen}(\resdom) = \cres^2\V\Capa(\resdom)$ is given by $$ w_0 = \begin{pmatrix}
        1\\
        1
    \end{pmatrix} \text{ and } w_1 = \begin{pmatrix}
        -1\\
        1\end{pmatrix}, $$
        with $$\lambda_0 = \Capa^{gen}(\resdom)_{1,1} + \Capa^{gen}(\resdom)_{2,2} \quad \text{ and } \quad \lambda_1 = \Capa^{gen}(\resdom)_{1,1} - \Capa^{gen}(\resdom)_{2,2}$$ being the respective eigenvalues.
    It follows that when plugging $v_0 = a w_0$ or $v_1 = a w_1$ into the system \eqref{eq:char_resonance_num}, one obtains that
    \begin{equation*}
        \begin{split}
            \Capa^{gen}(\resdom)v_j - \omega_0^2\left( v_j - \beta \cres^2 \lvert a \rvert^2 v_j\right) \quad \quad  \quad \quad  \quad \quad  \quad \quad \quad \quad\\= \left(\lambda_j - \omega_0^2\left(1  - \beta \cres^2 \lvert a \rvert^2 \right)\right)v_j, \quad \text{ for } j \in \{0,1\}.
        \end{split}
    \end{equation*}
    It follows that $(v_j, \omega_j(a))$ with 
    $$(\omega_j(a))^2 = \frac{\lambda_j}{\left(1  - \beta \cres^2 \lvert a \rvert^2 \right)}$$
    solves \eqref{eq:char_resonance_num}.
\end{proof}

\begin{remark}
    In the symmetric dimer case, the leading order asymptotics of the resonant solutions to the linear system are preserved for all amplitudes of the nonlinearity given in Remark \ref{rem:different_nonlin}. Comparing equations \eqref{eq:res_monomer} and \eqref{eq:resonance_dimer}, one can observe that the dependence of the leading order of the subwavelength resonances $\omega(a)$ on the amplitude $a$ is reminiscent of the monomer case discussed in Remark \ref{rem:different_nonlin}.
\end{remark}

Of course, in the nonlinear setting with a cubic nonlinearity, one expects a third solution. For this solution to occur, the magnitude of the solution modes needs to be sufficiently large. In Figure \ref{fig:resonancec_plot_1}, the solutions $q_0$ to equation \eqref{eq:char_resonance_nonlin} are displayed for a domain $\resdom$ in the context of Lemma \ref{lem:non-lin_sol}. Concretely, the solutions $(q_0,\omega_0)$ of equation \eqref{eq:char_resonance_num} are computed for $B_{r_1} = B_{0.2}(-\frac{1}{2}e_3)$ and $B_{r_2} = B_{0.2}(\frac{1}{2}e_3)$, where $B_r(x)$ for some $r > 0$ and $x \in \mathbb{R}^3$ is defined as $B_r(x) = \{ y \in \mathbb{R}^3 : \lvert x - y\rvert < r\}$. The material parameters are chosen as follows $\cres = 1$ and $\beta = -0.1i$. 

Due to the complex nature of the solutions, we decided to plot the magnitude of each entry $\lvert(q_0)_1\rvert$, $\lvert(q_0)_2\rvert$ and display the phase of $(q_0)_1/(q_0)_2$ as the respective color of the curve.

The solutions predicted from Lemma \ref{lem:non-lin_sol} can be seen as the diagonal line in Figure \ref{fig:solution_plot_1} (the predicted solutions lie on top of each other). For a sufficiently big magnitude, one sees another solution curve occurring. The amplitudes of the entries $\lvert(q_0)_1\rvert$ and $\lvert(q_0)_2\rvert$ are symmetric and the phase difference is conjugate, as expected from a symmetric system. This is made more precise in the following lemma.

\begin{figure}[htbp]
    \centering
    \begin{subfigure}[h]{0.32\textwidth}
        \centering
        \includegraphics[width = \textwidth]{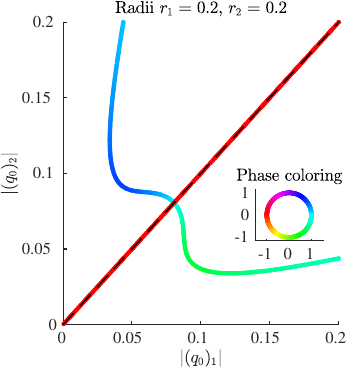}
        \caption{}
        \label{fig:solution_plot_1}
    \end{subfigure}\hfill
    \begin{subfigure}[h]{0.32\textwidth}
        \centering
        \includegraphics[width = \textwidth]{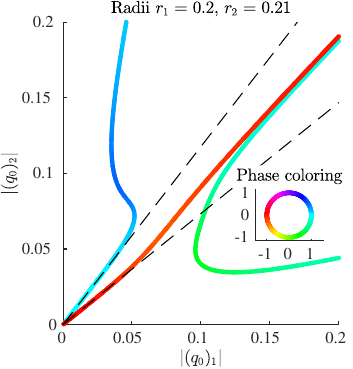}
        \caption{}
        \label{fig:solution_plot_2}
    \end{subfigure}\hfill
    \begin{subfigure}[h]{0.32\textwidth}
        \centering
        \includegraphics[width = \textwidth]{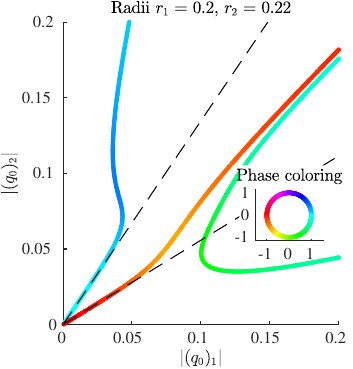}
        \caption{}
        \label{fig:solution_plot_3}
    \end{subfigure}
    \caption{In these plots the solution curves of equation \eqref{eq:char_resonance_num} for different choices of $\resdom$ are plotted. From left to right, $\resdom = B_{0.2}\left(-\frac{1}{2}e_3\right) \cup B_{0.2}\left(\frac{1}{2}e_3\right)$, $\resdom = B_{0.2}\left(-\frac{1}{2}e_3\right) \cup B_{0.21}\left(\frac{1}{2}e_3\right)$ and $\resdom = B_{0.2}\left(-\frac{1}{2}e_3\right) \cup B_{0.22}\left(\frac{1}{2}e_3\right)$, respectively. The remaining parameters are given by $\beta = -0.1i\left|B_{0.2}\left(-\frac{1}{2}e_3\right)\right|^{-2}$ and $\cres = 1$. The colorful curves trace the magnitude of the entries of different solutions $q_0$ to equation \eqref{eq:char_resonance_num} (the respective $\omega_0$ can be found in Figure \ref{fig:resonancec_plot}). The colors account for the complex valued solutions $q_0 \in \mathbb{C}^2$ and depict the phase of $(q_0)_1 / (q_0)_2$. The black dashed curves are the solutions to the respective linear problem when $\beta = 0$.}
    \label{fig:solution_plot}
\end{figure}

\begin{lemma}
    Let $\resdom$ be a smooth bounded domain with well-separated connected components $B_1,~ B_2$ with connected boundaries.
   Suppose that $\resdom = B_1 \cup B_2$ is preserved under the mirror symmetry $ P(x_1,x_2,x_3) = (x_1,x_2,-x_3) $ and $P(B_1) = P(B_2)$. Then, if $(\omega_0,q_0)$ is a solution of \eqref{eq:char_resonance_num}, then 
    $$ \left(\omega_0, p_0 \right) \quad \text{ with } \quad p_0 = \begin{pmatrix}
        (q_0)_{2}\\
        (q_0)_{1}
    \end{pmatrix} $$
    is also a solution of \eqref{eq:char_resonance_num} and the phase of $(p_0)_1/(p_0)_2$ is conjugate to the phase of $(q_0)_1/(q_0)_2$.
\end{lemma}

\begin{proof} When $\resdom$ satisfies the above symmetry, it follows that $\Capa^{gen}(\resdom)$ is symmetric.
    Plugging $(\omega_0,p_0)$ into equation \eqref{eq:char_resonance_num}, one observes that the first component of \eqref{eq:char_resonance_num} is given by the second component when plugging in $(\omega_0,q_0)$ and the second is equal to the first. Thus, both are equal to zero and $(\omega_0,p_0)$ is a solution to \eqref{eq:char_resonance_num}.
\end{proof}

In Figure \ref{fig:solution_plot_2} and Figure \ref{fig:solution_plot_3}, the solutions for asymmetric systems are computed. The domain $B_2$ is changed to $B_{0.21}(\frac{1}{2}e_3)$ and $B_{0.22}(\frac{1}{2}e_3)$, respectively. For $\lvert(q_0)_1\rvert$, $\lvert(q_0)_2\rvert$ small enough, one observes that the solution curves are close to the solutions of the linear setting, as one would expect from linearization techniques. Due to the asymmetry of the system, the solutions $q_0$ of the linear system have entries of different magnitudes $\lvert(q_0)_1\rvert \not= \lvert(q_0)_2\rvert$, which makes the two subspaces visible (light blue and red curves no longer lie on top of each other, as in Figure \ref{fig:solution_plot_1}). For sufficiently large magnitudes, a third solution appears. In the symmetric setting, the third solution did not interfere with the `linear' solutions predicted by Lemma \ref{lem:non-lin_sol}. However, when the symmetry is broken, the third solution interferes with the first two. This can be seen in the splitting of the light blue curve.

The leading order asymptotics of the respective subwavelength resonances of the nonlinear setting are displayed in Figure \ref{fig:resonancec_plot}. The leading order asymptotics of the subwavelength resonances of the symmetric setting associated to Figure \ref{fig:solution_plot_1} are depicted in Figure \ref{fig:resonancec_plot_1}. In the main plot, the value of the leading order asymptotics $\omega_0$ of the  subwavelength resonances are plotted in the complex plane. In the inlet plot the leading order asymptotics $q_0$ of the subwavelength solution with the colors of the different solution curves corresponding to the different curves of the leading order asymptotics of the subwavelength resonances are plotted.
In Figure \ref{fig:resonancec_plot_1}, one observes that the curve of the leading order asymptotics of the subwavelength resonances which are associated to the third solution curve, touches the curve of the leading order asymptotics of the subwavelength resonances associated to the solution $q_0 = (a,a)^\top, a \in \mathbb{C}$, where the superscript $\top$ denotes the transpose. It is precisely that curve that is mixing with the nonlinear solutions for sufficiently large amplitudes in the asymmetric settings, as can be seen from Figures \ref{fig:solution_plot_2} and \ref{fig:resonancec_plot_2}, and Figures \ref{fig:solution_plot_3} and \ref{fig:resonancec_plot_3}, respectively.

\begin{figure}[htbp]
    \centering
    \begin{subfigure}[h]{0.32\textwidth}
        \centering
        \includegraphics[width = \textwidth]{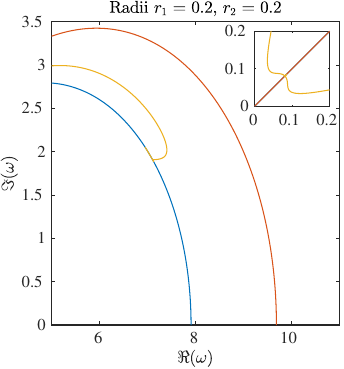}
        \caption{}
        \label{fig:resonancec_plot_1}
    \end{subfigure}\hfill
    \begin{subfigure}[h]{0.32\textwidth}
        \centering
        \includegraphics[width = \textwidth]{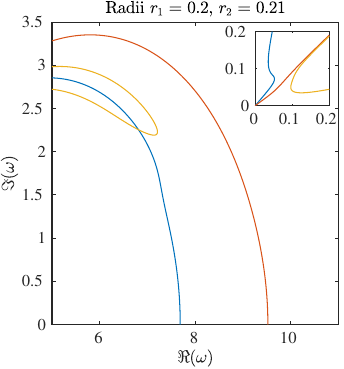}
        \caption{}
        \label{fig:resonancec_plot_2}
    \end{subfigure}\hfill
    \begin{subfigure}[h]{0.32\textwidth}
        \centering
        \includegraphics[width = \textwidth]{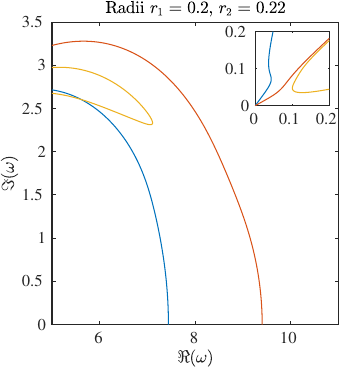}
        \caption{}
        \label{fig:resonancec_plot_3}
    \end{subfigure}
    \caption{In this figure the leading order asymptotics of the subwavelength resonances associated to the settings discussed in the caption of Figure \ref{fig:solution_plot} are displayed. The leading order asymptotics of the resonances are displayed as curves in the complex plane. The color of each curve is associated to the color of a solution curve $q_0$ depicted in the inlet plots.}
    \label{fig:resonancec_plot}
\end{figure}

\section{Conclusion and outlook}\label{section:outlook}
In this paper, we have provided for the first time a discrete approximation to the nonlinear resonance problem for finite systems of subwavelength resonators. Our discrete approximation allows us not only to derive the effect of weak nonlinearities on the linear resonances but more importantly characterizes the solutions which are induced by stronger nonlinearities. 

It would be very interesting, firstly, to rigorously derive modal decompositions for scattered solutions by nonlinear systems of resonators, and secondly, {to prove existence of soliton-like resonances and their associated localized eigenmodes and analyze their stability.} 
{Another very interesting problem is to develop
a Foldy-Lax formulation or a point-interaction approximation  for nonlinear subwavelength resonators and study their limits in large, dilute systems as done in the linear case,  for instance, in \cite{mourad,haizhang}.}
Our results in this paper can be extended to periodic systems of resonators. In forthcoming works, {we plan on the one hand to prove existence of soliton-like 
solutions in periodic structures that are induced by nonlinear defects into a linear structure} and on the other hand, to elucidate the interplay between nonlinearities and subwavelength resonances, in particular the nonlinearity-induced topological transitions in such systems.

\section*{Acknowledgments}
This work was supported in part by the Swiss National Science Foundation grant number 200021--200307. T.K. thanks Tianwei Yu, Wouter Tonnon and Michael Weinstein for the interesting discussions, and Erik Hiltunen for the helpful feedback.

\bibliography{references.bib}

\end{document}